\documentclass[a4paper,onecolumn,11pt,accepted=2026-01-19]{quantumarticle}
\pdfoutput=1

\newcommand{\highlighted}{1}
\ifdefined\highlighted
    \usepackage[normalem]{ulem}

    \newcommand{\remove}[1]{{\color{red}\stkout{#1}}}
\else

    \newcommand{\remove}[1]{}
\fi

\usepackage{hyperref}

\usepackage{graphicx}%
\usepackage{multirow}%
\usepackage{amsmath,amssymb,amsfonts}%
\usepackage{amsthm}%
\usepackage{mathrsfs}%
\usepackage[title]{appendix}%
\usepackage{xcolor}%
\usepackage{textcomp}%
\usepackage{manyfoot}%
\usepackage{booktabs}%
\usepackage{ulem}% for \sout 
\usepackage{algorithm}%
\usepackage{algorithmicx}%
\usepackage{algpseudocode}%
\usepackage{listings}%
\usepackage[numbers]{natbib}
\usepackage{bm}
\usepackage{arydshln} % for \hdashline
\usepackage[makeroom]{cancel} % for the \cancel command

\usepackage{fontenc}
\usepackage{mathtools}
\usepackage{bbm}
\usepackage{physics}
\usepackage[bb=boondox]{mathalfa}
\usepackage{ytableau}
\usepackage{url}
\usepackage{soul}

\usepackage{style}

\usepackage[noabbrev]{cleveref} % for cref
\newcommand{\U}{\mathrm{U}}
\newcommand{\Ort}{\mathrm{O}}

\newcommand{\T}{\mathcal{T}}

\theoremstyle{plain}
\newtheorem{theorem}{Theorem}%
\newtheorem{proposition}{Proposition}% 
\newtheorem{lemma}{Lemma}%

\theoremstyle{definition}
\newtheorem{definition}{Definition}%
\newtheorem{example}{Example}%
\newtheorem{remark}{Remark}%

\begin{document}

\title{Generalized group designs: constructing novel unitary 2-, 3- and 4-designs}

\author[1,2]{\'Agoston Kaposi}
\email{kaposiagoston@inf.elte.hu}

\author[1,2]{Zolt\'an Kolarovszki}
\email{kolarovszki.zoltan@wigner.hun-ren.hu}

\author[1,3]{Adri\'an Solymos}
\email{solymos.adrian@wigner.hun-ren.hu}

\author[1,2,4,5]{Zolt\'an Zimbor\'as}
\email{zimboras.zoltan@wigner.hun-ren.hu}

\affil[1]{Quantum Computing and Quantum Information Research Group, HUN-REN Wigner Research Centre for Physics, Konkoly–Thege Mikl\'os \'ut 29-33, Budapest, H-1525, Hungary}

\affil[2]{Department of Programming Languages and Compilers, E\"otv\"os  Lor\'and  University, P\'azmány P\'eter s\'et\'any 1/C, Budapest, H-1117, Hungary}

\affil[3]{Department of Physics of Complex Systems, E\"otv\"os  Lor\'and  University, P\'azmány P\'eter s\'et\'any 1/A, Budapest, H-1117, Hungary}
\affil[4]{Algorithmiq Ltd, Kanavakatu 3C, Helsinki, 00160, Finland}
\affil[5]{University of Helsinki, Yliopistonkatu 4, Helsinki, 00100, Finland}

\maketitle

\begin{abstract}
    Unitary designs are essential tools in several quantum information protocols. Similarly to other design concepts, unitary designs are mainly used to facilitate averaging over a relevant space, in this case, the unitary group $\U(d)$.  While it is known that exact unitary $t$-designs exist for any degree $t$ and dimension $d$, the most appealing type of designs, group designs (in which the elements of the design form a group), can provide at most $3$-designs. Moreover, even group $2$-designs can exist only in limited dimensions.
    In this paper, we present novel construction methods for creating \textit{exact generalized group designs} based on the representation theory of the unitary group and its finite subgroups that overcome the $4$-design-barrier of unitary group designs. Furthermore, a construction is presented for creating generalized group $2$-designs in arbitrary dimensions.
\end{abstract}

\section{Introduction}
    Unitary $t$-designs were first introduced in Refs.~\cite{dankert2005efficient, Dankert_2009},
    and subsequently have become an essential tool within the field of quantum information science.
    These finite collections of $d \times d$ unitary matrices possess a unique property: By twirling (i.e., averaging) a certain matrix with the $t$-fold tensor products of these unitary matrices, the result is equal to the twirling over the \textit{entire} unitary group $\U(d)$ with respect to the Haar measure.
    This property makes them advantageous to use in protocols requiring unitaries sampled randomly from the Haar measure, thus potentially decreasing the resources needed, as generic Haar random unitaries are hard to implement.
    Hence, unitary designs have wide-ranging applications in quantum information science, in particular for quantum process tomography and channel fidelity estimation \cite{Scott_2008,Dankert_2009}, the construction of unitary codes~\cite{unitary_designs_and_codes}, derandomization of probabilistic constructions~\cite{Gross_2014}, entanglement detection~\cite{Bae_2019}, the study of quantum chaos~\cite{roberts2017chaos,leone2021quantum}, randomized benchmarking~\cite{wallman2014randomized,jafarzadeh2020randomized}, and shadow estimation~\cite{huang2020predicting,chen2021robust,arienzo2022closed,helsen2022thrifty}.

    One of the most elegant methods for constructing unitary designs involves delving into the representation theory of finite groups. The collection of unitaries derived from any unitary irreducible representation of a finite group constitutes a $1$-design, and given additional properties, these elements can potentially form a higher-degree unitary design~\cite{evenly_distributed_unitaries}.
    We call such collections of unitaries that form a group (such as those coming from representations) \textit{unitary group designs}\footnote{We note that the term group design is also used in the literature to denote a unitary ensemble that approximates the Haar measure restricted to a subgroup of the unitary group, see, e.g., Refs.~\cite{west2025nogotheoremssublineardepthgroup,grevink2025glueshortdepthdesignsunitary}.}.
    Notably, some of the most renowned families of unitary $t$-designs fall under the category of group designs: The Clifford groups over finite fields, both in their multipartite and single-particle variants, provide elegant and practical $2$-designs in prime power dimensions~\cite{Chau}. Consequently, they find applications in randomized benchmarking for both qubit and qudit systems~\cite{wallman2014randomized, jafarzadeh2020randomized, heinrich2023randomizedbenchmarkingrandomquantum}. The multipartite Clifford group for qubit systems even constitutes a unitary $3$-design~\cite{Zak,zhu2017multiqubit.96.062336,kueng2015qubit}, proving to be advantageous, for instance, in single-shot shadow estimation~\cite{chen2021robust,arienzo2022closed,helsen2022thrifty}.

    The rotten apple that spoils the barrel is the fact that, in dimensions greater than $2$, there do not exist group $4$-designs (and thus higher-degree designs)~\cite{bannai}. Moreover, one can construct group $2$-designs only for a few non-prime-power dimensions (in particular, one can define modular versions of the Clifford group for non-prime-power dimensions, however, these fail to be 2-designs~\cite{graydon2021clifford}). This severely limits the use of group designs for tasks such as randomized benchmarking in these dimensions.

    Constructing unitary $t$-designs using finite group representations will eventually fail because the $t$-fold tensor product of the chosen representation will decompose nontrivially on some of the $\U(d)$-irreducible subspaces of the $t$-fold tensor product space.
    Thus, our research was mainly focused on the behavior of the group representations on those $\U(d)$-irreducible subspaces.
    In this paper, we present novel construction methods for creating \textit{generalized group designs} based on such representation theoretic considerations. This allows us to build higher-degree designs from lower ones and also enables us to define a procedure for constructing $2$-designs in arbitrary dimensions.

    The structure of the paper is as follows: Section~\ref{sec:background} contains necessary basic definitions regarding unitary designs;
    Section~\ref{sec:higher_degree} provides a construction of $t$-designs from finite unitary subgroups and presents some examples for the construction; Section~\ref{sec:2design_arbitrary} describes an explicit construction for generalized group $2$-designs in arbitrary dimensions; finally, Section~\ref{sec:constructing_unitary_design_from_orthogonal} outlines a method for constructing unitary $2$- and $3$-designs from orthogonal $2$- and $3$-designs, respectively.

\section{Preliminaries} \label{sec:background}
    This section provides the essential definitions and statements employed in the proposed constructions. A basic familiarity with representation theory is assumed; a brief recap on the basic concepts is provided in Appendix~\ref{app:representation_theory}.
    It is important to note, that this paper is mainly concerned with the construction of exact designs, but we provide a statement for the construction of unitary $2$- or $3$-designs from orthogonal ones, which holds in the approximate case as well. For concrete constructions regarding approximate designs, see Refs.~\cite{Dankert_2009,Brand_o_2016,oszmaniec2022logic,haferkamp2022random}.

    Several different but equivalent definitions can be found for unitary designs and group designs in the literature~\cite{seymour1984averaging,Bengtsson_Życzkowski_2017,evenly_distributed_unitaries}. However, the notion of unitary $t$-designs, where the summation over the finite set is weighted with a certain weight function, is crucial for this paper since the constructions presented are generally of this type.
    \begin{definition}[Unitary $t$-design] \label{def:t-design}
        A finite set $\mathcal{V} \subset \U(d)$ with weight function $w: \mathcal{V} \to [0, 1]$ is called a unitary $t$-design if the following equation holds for any linear transformation $M$ on $(\C^{d})^{\otimes t}$:
        \begin{equation}\label{eq:t_design_with_weights}
             \sum_{V \in \mathcal{V} } w(V) \,
             V^{\otimes t}M\left(V^{\otimes t}\right)^{\dagger}
             =
             \integral{U\in \U(d)}{} U^{\otimes t}M\left(U^{\otimes t}\right)^{\dagger} \mathrm{d}U,
        \end{equation}
        where the integral on the right-hand side is taken over all elements in $\U(d)$ with respect to the Haar measure. The number $t$ is called the degree of the design. 
    \end{definition}
    In this definition and the rest of the paper $\U(d)$ can be naturally identified with its defining representation.
    More concretely, $U^{\otimes t} = \Pi_{\Yt{1}}(U)^{\otimes t}$ for $U \in \U(d)$, where $\Pi_{\Yt{1}}$ is the defining representation of $\U(d)$.
    Moreover, the weight function of a $t$-design $\mathcal{V}$ is the constant function $w \equiv 1 / \abs{\mathcal{V}}$ unless otherwise stated. Finally, a unitary $t$-design is called a \textit{unitary group $t$-design} when it possesses a group structure. In our proposals, however, we will construct $t$-designs which are generalizations of group $t$-designs in the sense that they are constructed from products of finite groups:
    \begin{definition}[Generalized group $t$-design]\label{def:generalized_group_design}
        Let $\mathcal{V}_1,\dots,\mathcal{V}_\ell < \U(d)$ be finite subgroups. We call $\mathcal{V} \coloneqq \mathcal{V}_1\cdot\ldots\cdot \mathcal{V}_\ell$ a generalized group $t$-design if $\mathcal{V}$ forms a unitary $t$-design with weight function given by
        \begin{equation}\label{eq:weights_for_multiple_sets}
            w (U) \coloneqq 
            \frac
              {\left| \left\{ 
                (h_1, \dots, h_\ell) \in \mathcal{V}_1\times\dots\times \mathcal{V}_\ell : \prod_{j=1}^{\ell}h_j  = U
              \right\} \right|}
              {\abs{\mathcal{V}_1}\cdots \abs{\mathcal{V}_\ell}}.
        \end{equation}
    \end{definition}

    \begin{remark}
        The above careful definition of the weight function is needed as, in general, the product of finite subgroups of a group does not form a set whose cardinality is equal to the product of the cardinalities of the constituent subgroups.
    \end{remark}

\section{Constructing higher-degree generalized group designs} \label{sec:higher_degree}

    Construction of unitary designs for general $t$ and $d$ is a mathematically challenging task \cite{BANNAI2022108457}.
    A well-known example of a unitary design is the Clifford group, which constitutes a group $3$-design for qubit systems~\cite{zhu2017multiqubit.96.062336,Zak,kueng2015qubit}, but this specific example does not generalize for higher-degree designs. Fortunately, using the representation theory of finite groups, one may discover group $t$-designs with $t=1,2,3$ for specific dimensions.
    However, the construction of higher-degree group designs based on representation theory is impossible because of the nonexistence of group $4$-designs in dimensions greater than $2$~\cite{bannai}. The resulting $4$-design-barrier for group designs compelled us to search for higher-degree designs that are not group designs.
    In this section, we present a construction that pushes the $4$-design-barrier further for unitary designs by introducing an extended notion of group designs, which we will refer to as generalized group designs (see Def.~\ref{def:generalized_group_design}).

    First, another characterization of unitary $t$-designs (equivalent to Def.~\ref{def:t-design}) is presented based on representation theory (see also, Appendix of \cite{evenly_distributed_unitaries}). Consider the $t$-fold tensor product of the defining representation of $\U(d)$:  the underlying vector space  $(\C^{d})^{\otimes t}$ splits up under the action of $\U(d)$ into the different isotypic components labeled by Young diagrams as
    \begin{equation}
        (\C^{d})^{\otimes t} = \bigoplus_{\gamma \in \Gamma } \; \mcK_{\gamma} \otimes \mcH_{\gamma},
    \end{equation}
    where $\Gamma$ is the set of Young diagrams with $t$ number of boxes and at most $d$ rows, $\mcH_{\gamma}$ carries the $\U(d)$-irrep labeled by the Young diagram $\gamma$ and $\mcK_{\gamma}$ is the multiplicity space where $\U(d)$ acts trivially \cite{FH}. Let us denote by $P_{\gamma}$
    the projection corresponding to the $\mathcal{W}_\gamma \coloneqq \mcK_{\gamma} \otimes \mcH_{\gamma}$ subspace. As previously mentioned, we present another equivalent characterization of unitary $t$-designs, which may be easily justified using Schur's lemma:   
    \begin{proposition}\label{the_proposition}
        A finite set $\mathcal{V}\subset \U(d)$ and a weight function $w: \mathcal{V} \to [0, 1]$ forms a unitary $t$-design if and only if the following equation is true for all linear transformations $M$ on $(\mathbb{C}^d)^{\otimes t}$:
        \begin{equation} \label{eq:the_proposition_eq}
            \sum_{V \in \mathcal{V} } w\left(V\right) V^{\otimes t} M\left(V^{\otimes t}\right)^{\dagger} = 
            \bigoplus_{\gamma\in\Gamma}
            \frac{\Tr_{\mcH_{\gamma}} (P_\gamma M P_\gamma)}{\dim \mcH_\gamma}  \otimes \mathbbm{1}_{\mcH_\gamma},
        \end{equation}
        where we used the notation as before, and $\Tr_{\mcH_{\gamma}}$ is the partial trace over $\mcH_{\gamma}$ of operators supported on the subspace  $\mathcal{W}_\gamma$.
    \end{proposition}

    Based on Proposition~\ref{the_proposition} the condition for a finite subgroup of the unitary group to be a unitary $t$-design can be reformulated to require that the appropriate irreducible subspaces of the unitary group remain irreducible when restricted to this finite subgroup.
    In light of this, the 4-design-barrier can be formulated in the following way: there is no finite subgroup of $\U(d)$ such that all irreps of $\U(d)$ labeled by Young diagrams with 4 boxes, restricted to this subgroup, remain irreducible if $d>2$.
    The idea behind our construction is that for each Young diagram with $t$ boxes and at most $d$ rows we find a finite subgroup of $\U(d)$, such that the irrep of $\U(d)$ corresponding to this Young diagram remains irreducible when restricted to this subgroup.
    From these irreps, one can attempt to create a generalized unitary group $t$-design for $t\geq 4$.
    The fly in the ointment is that the irreducible subspace corresponding to such irreps might appear elsewhere in the irreducible decomposition of the $t$-fold tensor product of $\U(d)$ upon restriction. This might lead to a non-trivial intertwiner between these irreducible subspaces when averaging, ruining the construction. To rescue this idea, we aim to exclude these situations, in a carefully crafted theorem.
    \begin{theorem} \label{thm:creating_t_design_from_groups}
        Let $\varPi_{\Yt{1}}$ denote the defining representation of the unitary group $\U(d)$. Let $\Gamma$ denote the set of all Young diagrams with t boxes and at most d rows, corresponding to the irreducible representations appearing in the irreducible decomposition of $\varPi_{\Yt{1}}^{\otimes t}$.
        Let $G_1, \dots, G_\ell < \U(d)$ be finite unitary subgroups such that:
        \begin{enumerate}
            \item For each Young diagram $\gamma\in\Gamma$, there is a $j \in \left\{1, \dots, \ell\right\}$ such that the representation $\varPi_{\gamma}|_{G_j}$ remains an irreducible representation;
            \item For any $\gamma, \eta \in \Gamma$, there is a $j \in \left\{1, \dots, \ell\right\}$ such that there is no non-trivial intertwiner between the representations $\varPi_\gamma|_{G_j}$ and $\varPi_\eta|_{G_j}$.
        \end{enumerate}
        Then the product $G_1 \cdot \ldots \cdot G_\ell$ forms a generalized group $t$-design.
    \end{theorem}
    \begin{proof}
        Let us denote twirling with a subgroup $G < \U(d)$ as
        \begin{equation}\label{eq:twirl_by_G}
            \T_{G} (X) \coloneqq \int_G U^{\otimes t} X (U^{\otimes t})^\dagger \, \dd U,
        \end{equation}
        where the integration is done over the Haar measure of $G$.
        Moreover, let us denote the consecutive twirls with $G_1, \dots, G_\ell$ by
        \begin{equation}
            \T \coloneqq \T_{G_1} \dots \T_{G_\ell},
        \end{equation}
        and let us have the following notations
        \begin{equation}
            \mathcal{P}_{\gamma,\eta}(X) \coloneqq P_{\gamma} X P_{\eta}, \hspace{2cm}
            \mathcal{P}_\gamma \coloneqq \mathcal{P}_{\gamma,\gamma},
        \end{equation}
        where $X$ is a linear operator on $(\C^d)^{\otimes t}$.
        Considering a fixed linear operator $M$ on $(\C^d)^{\otimes t}$,
        we want to show that
        \begin{equation}
            \T(M) =
            \sum_{\gamma \in \Gamma}
                \frac{\Tr_{\mathcal{H}{\gamma}} (\mathcal{P}_{\gamma} (M))}{\dim \mathcal{H}_{\gamma}}
                \otimes
                \mathbbm{1}_{\mathcal{H}_{\gamma}}.
        \end{equation}
        For this, let us decompose $\T(M)$ as
        \begin{equation}
            \T(M) = \sum_{\gamma, \eta \in \Gamma} \mathcal{P}_{\gamma,\eta} \T(M)
            = 
            \sum_{\gamma \in \Gamma} \mathcal{P}_{\gamma} \T(M) + \sum_{\substack{\gamma, \eta \in \Gamma \\  \gamma \neq \eta}}
            \mathcal{P}_{\gamma,\eta} \T (M)\text{.}
        \end{equation}
        Let $\gamma, \eta \in \Gamma$.
        Firstly, we prove that $\mathcal{P}_{\gamma,\eta} \T (M) = 0$ if $\gamma \neq \eta$.
        We can use the fact that for all $\gamma, \eta \in \Gamma$ and any subgroup $G < \U(d)$
        \begin{equation}
            \mathcal{P}_{\gamma, \eta} \T_{G} = \mathcal{P}_{\gamma, \eta} \T_{G} \mathcal{P}_{\gamma, \eta},
        \end{equation}
        which is trivial by using Schur's lemma.
        Hence, the following is true:
        \begin{equation}\label{eq:ptp}
            \mathcal{P}_{\gamma,\eta} \T (M) = \mathcal{P}_{\gamma, \eta} \T_{G_1} \mathcal{P}_{\gamma, \eta} \dots \mathcal{P}_{\gamma, \eta} \T_{G_\ell} \mathcal{P}_{\gamma, \eta}(M)\text{.}
        \end{equation}

        Let $j$ be the index for which there is no non-trivial intertwiner between $\varPi_{\gamma}|_{G_j}$ and $\varPi_{\eta}|_{G_j}$ as required by the second condition.
        It is easy to see that $\T_{G_j}(X)$ is an intertwiner of the $t$-fold diagonal action of $G_j$ for any linear operator $X$ on $(\C^d)^{\otimes t}$. Furthermore, given that $G_j$ is a subgroup of the unitary group, $P_{\gamma}\T_{G_j}(X) P_{\eta}$ is also an intertwiner for the same representation as the product of intertwiners is also an intertwiner. Given that $P_{\gamma}\T_{G_j}(X) P_{\eta}$ would be an intertwiner exactly between the subrepresentations $\varPi_{\gamma}|_{G_j}$ and $\varPi_{\eta}|_{G_j}$ (with multiplicities), it has to be zero by the second condition. Moreover, because of the linearity of the above equation, if it transforms into the zero map at any point it will stay zero after subsequent twirls.

        Secondly, considering the case when $\gamma = \eta$, we can write Eq.~\eqref{eq:ptp} as
        \begin{align}
            P_{\gamma} \T (M) P_{\gamma} &=
            \mathcal{P}_{\gamma} \T_{G_1} \mathcal{P}_{\gamma}  \dots \T_{G_\ell} \mathcal{P}_{\gamma}(M).
        \end{align}
        By the first assumption in the Theorem, for each $\gamma \in \Gamma$, there exists a group $G_j$ such that $\varPi_{\gamma} |_{G_j}$ is irreducible. Note that from Proposition \ref{the_proposition} it follows that:
        \begin{equation}
            \mathcal{P}_\gamma \T_{G_j}(X)=\mathbb{0} \oplus \frac{\Tr_{\mcH_{\gamma}} (      \mathcal{P}_{\gamma} (X)
                )
            }{
                \dim \mcH_{\gamma}
            } \otimes
                \mathbbm{1}_{\mcH_{\gamma}}\text{,}
        \end{equation}
        where $X$ is any linear transformation on $(\mathbb{C}^d)^{\otimes t}$ and where we denoted the zero operator on the orthogonal complement of $P_\gamma$ by $\mathbb{0}$.
        
        As a consequence, it is clear that $\mathcal{P}_\gamma \T_{G_j} \T_{G_k} = \mathcal{P}_\gamma \T_{G_j}$ for $k = j+1, \dots, \ell$ since
        \begin{align}
        \begin{split}
            \mathcal{P}_\gamma \T_{G_j} \T_{G_k} (X)
            = \mathbb{0} \oplus
                \frac{
                \int_{G_k}\Tr_{\mcH_{\gamma}} \left( 
                    [\mathbbm{1}_{\mcK_{\gamma}}\otimes \varPi_{\gamma}(U)]
                        \mathcal{P}_{\gamma} (X)
                    [\mathbbm{1}_{\mcK_{\gamma}}\otimes \varPi_{\gamma}(U)^{\dagger}]
                \right)\,\dd U
            }{
                \dim \mcH_{\gamma}
            }
            \otimes
                \mathbbm{1}_{\mcH_{\gamma}},
        \end{split}
        \end{align}
        where the integration is done over the Haar measure of $G_k$. The partial trace is cyclic on the vector space it is tracing over, i.e., $\mcH_\gamma$, and the operators on $\mcK_\gamma$ commute. Therefore this partial trace can be treated as cyclic to get
        \begin{align} \label{eq:previous_fact}
        \begin{split}
            \mathcal{P}_\gamma \T_{G_j} \T_{G_k} (X) =
            \mathbb{0} \oplus
            \frac{
                 \Tr_{\mcH_{\gamma}} (
                        \mathcal{P}_{\gamma} (X)
                )
            }{
                \dim \mcH_{\gamma}
            } \otimes
                \mathbbm{1}_{\mcH_{\gamma}} = \mathcal{P}_\gamma \T_{G_j} (X).
        \end{split}
        \end{align}

        Moreover, the twirlings $\T_{G_k}$ for $k = 1,\,\dots, \,j-1$ leave $\T_{G_j}(X)$ invariant, since the above operator is invariant to the adjoint action of any subgroup.
        More concretely, let $k \in \left\{1,\dots,j-1\right\}$. Then we can write for any linear transformation $X$ on $(\mathbb{C}^d)^{\otimes t}$
        \begin{align}
        \begin{split}
            \mathcal{P}_{\gamma} & \T_{G_k} \mathcal{P}_\gamma \T_{G_j} \mathcal{P}_\gamma ( X )
            =
            \mathcal{P}_{\gamma} \T_{G_k} \left(
                \mathbb{0} \oplus \frac{\Tr_{\mcH_{\gamma}} (\mathcal{P}_{\gamma} (X) )}{\dim \mcH_{\gamma}}
                    \otimes \mathbbm{1}_{\mcH_{\gamma}}
            \right)
            \\&=
            \mathbb{0} \oplus \int_{G_k} (\mathbbm{1}_{\mcK_{\gamma}} \otimes \varPi_{\gamma}(U))  \left(
                \frac{\Tr_{\mcH_{\gamma}} (\mathcal{P}_{\gamma} (X))}{\dim \mcH_{\gamma}}
                    \otimes \mathbbm{1}_{\mcH_{\gamma}}
            \right)
            (\mathbbm{1}_{\mcK_{\gamma}} \otimes \varPi_{\gamma}(U))^\dagger \dd U
            \\&=
            \mathbb{0} \oplus
                \frac{\Tr_{\mcH_{\gamma}} (\mathcal{P}_{\gamma} (X))}{\dim \mcH_{\gamma}}
                    \otimes \mathbbm{1}_{\mcH_{\gamma}},
        \end{split}
        \end{align}
        hence, $\mathcal{P}_\gamma \T_{G_k} \T_{G_j} = \mathcal{P}_\gamma \T_{G_j}$. By combining this with the fact that $$
        \mathcal{P}_\gamma \T_{G_j} \T_{G_k} = \mathcal{P}_\gamma \T_{G_j} \quad (k=j+1, \ldots, \ell)
        $$
        by Eq.~\eqref{eq:previous_fact}, 
        we get that
        \begin{equation}
            \mathcal{P}_{\gamma} \T = \mathcal{P}_{\gamma} \T_{G_j}.
        \end{equation}
        As a result, one gets
        \begin{equation}
            \T(M) =
            \sum_{\gamma \in \Gamma}
                \frac{\Tr_{\mathcal{H}{\gamma}} (\mathcal{P}_{\gamma} (M))}{\dim \mathcal{H}_{\gamma}}
                \otimes
                \mathbbm{1}_{\mathcal{H}_{\gamma}}.
        \end{equation}
        which proves the Theorem by Proposition~\ref{the_proposition}.
    \end{proof}
    
    Theorem~\ref{thm:creating_t_design_from_groups} enables us to search for examples of unitary $t$-designs using character theory. In particular, for some low-dimensional $2$- and $3$-design constructions a similar technique was used in Ref.~\cite{kaposi2023constructing}.
    To be explicit, given a finite unitary subgroup $G < \U(d)$, one can examine if the representation labeled by the Young diagram $\gamma$ with $t$ boxes and at most $d$ rows remains irreducible when restricted to $G$. Hence, one can search for a set of finite subgroups of $\U(d)$ such that the first condition is fulfilled.
    For the second condition, consider two Young diagrams $\gamma, \eta$.
    When $\dim \varPi_\gamma \neq \dim \varPi_\eta$, consider $\dim \varPi_\gamma > \dim \varPi_\eta$ without loss of generality. By the first condition, there is a finite subgroup $G$ for which the representation $\varPi_\gamma|_{G}$ remains irreducible, hence, the second condition is automatically fulfilled.
    When $\dim \varPi_\gamma = \dim \varPi_\eta$, further investigation is needed to verify that the representations are inequivalent when restricted to some finite subgroup $G$.
    To verify the irreducibility and inequivalence of the representations,
    the method described in Appendix~\ref{app:character_theory} is used.
    Using these methods, we have found the following examples of generalized group $4$-designs in the GAP system (version 4.14.0)~\cite{GAP4}:
    
    \begin{example}[A $6$-dimensional $4$-design]
        In this example, not all the groups are commonly referenced in the literature, hence, we employ the notation utilized within the GAP system to avoid ambiguity.
        Any of the $6$-dimensional irreducible representations of the group \texttt{"S3x6\_1.U4(3).2\_2"} or the group \texttt{"6.U6(2)M6"} remain irreducible on the subspaces corresponding to $\ytableausetup{aligntableaux=center}\Yt{3,1}, \Yt{2,2}, \Yt{2,1,1}, \Yt{1,1,1,1}$. Meanwhile, the $22$nd irreducible representation of the group \texttt{"2.J2"} is irreducible when restricted to the subspace corresponding to $\Yt{4}$.
        Finally, the irreps corresponding to $\Yt{2,2}$ and $\Yt{2,1,1}$ have the same dimension, but are inequivalent when restricted to the corresponding group, and thus, they also satisfy the second condition in Theorem~\ref{thm:creating_t_design_from_groups}.
    \end{example}

    \begin{example}[A $12$-dimensional $4$-design]\label{example:4design_12dim}
        The Suzuki sporadic simple group $\operatorname{Suz}$ has a central extension of order $6$ (denoted by \texttt{"6.Suz"} in GAP) which has an irreducible representation (the $153$rd in GAP) which remains irreducible on the subspaces corresponding to $\Yt{4}, \Yt{3,1}, \Yt{2,2}$ and $\Yt{2,1,1}$.
        Meanwhile, the standard representation of the alternating group on $13$ elements is irreducible on the subspace corresponding to $\Yt{1,1,1,1}$.
        Therefore, from these two unitary representations, a unitary $4$-design can be constructed in $12$ dimensions by Theorem~\ref{thm:creating_t_design_from_groups}.
    \end{example}

    Apart from the two examples presented so far, Theorem~\ref{thm:creating_t_design_from_groups} can also be used to find other examples, e.g. for $3$-designs. Without completeness, we present a few such examples here:

    \begin{example}[A $10$-dimensional $3$-design]\label{example:3design_10dim}
        The double cover of the Mathieu group $M_{22}$ (denoted by \texttt{2.M22} in GAP) has a $10$-dimensional representation which remains irreducible on subspaces corresponding to $\Yt{2,1}$ and $\Yt{1,1,1}$.
        Together with a $10$-dimensional representation of the double cover of the Mathieu group $M_{12}$ together with an outer automorphism of order $2$ (denoted by \texttt{2.M12.2} in GAP) which remains irreducible on the subspace corresponding to $\Yt{3}$. With these representations, a $10$-dimensional $3$-design can be constructed.
    \end{example}
    \begin{example}[A $13$-dimensional $3$-design]\label{example:3design_13dim}
        The non-split double cover of the projective symplectic group $\operatorname{PSp}_6(3)$ (denoted by \texttt{2.S6(3)} in GAP) has a $13$-dimensional representation which remains irreducible on subspaces corresponding to $\Yt{3}$ and $\Yt{2,1}$. Together with the standard representation of the alternating group on $14$ elements, the concatenated product of these groups constitute a $13$-dimensional $3$-design.
    \end{example}
    \begin{example}[A $18$-dimensional $3$-design]\label{example:3design_18dim}
        There is a central extension of the Janko group $\operatorname{J}_3$ by a cyclic group of order $3$ (denoted by \texttt{3.J3} in GAP) which has an $18$-dimensional representation which is already a $3$-design.
    \end{example}
    Note that in \Cref{example:4design_12dim,example:3design_10dim,example:3design_13dim,example:3design_18dim}, the dimensions of the subspaces are different, producing inequivalent representations on the subspaces.

\section{Constructing \texorpdfstring{$2$}{2}-designs in arbitrary dimension}\label{sec:2design_arbitrary}
    
    In prime power dimensions, the Clifford groups provide examples of unitary $2$-designs~\cite{wallman2014randomized,jafarzadeh2020randomized}.
    For arbitrary dimensions,
    an inductive construction for $t$-designs is provided in Ref.~\cite{BANNAI2022108457}.
    However, an explicit non-inductive construction for $2$-designs is not known for every dimensions.
    In this section, we provide an explicit construction of generalized group $2$-designs in arbitrary dimensions.

    \begin{lemma}\label{lemma:2design_symmetric_subspace}
        Consider a finite subgroup $G < \U(d)$, and consider the defining representation of $\U(d)$ restricted to this subgroup, denoted as $\varPi_{\Yt{1}} |_G$. Suppose that the restriction of the representation
        \(\varPi_{\Yt{1,1}}\) to \(G\) remains irreducible, and that
        \begin{equation}
            \varPi_{\Yt{2}}\big|_{G}
            \;\cong\;
            \varPi_{1} \oplus \varPi_{2},
        \end{equation}
        where \(\varPi_{1}\) and \(\varPi_{2}\) are two inequivalent irreducible
        representations of \(G\), each also inequivalent to
        \(\varPi_{\Yt{1,1}}\big|_{G}\).
        Let \(P_{1}\) and \(P_{2}\) denote the corresponding projectors onto the
        subrepresentations \(\varPi_{1}\) and \(\varPi_{2}\), respectively. Moreover, suppose that there exists $U \in \U(d)$ such that
        \begin{equation} \label{eq:desired_general}
            \Tr \left( P_1^{U^{\otimes 2}} P_1 \right)
            = \frac{(\Tr P_1)^2}{\Tr (P_1 + P_2)},
        \end{equation}
        then the product $G^{U} \cdot G$ forms a unitary $2$-design,
        where $G^{U} = \left\{ U g U^\dagger \;|\; g \in G \right\}$.
    \end{lemma}
    \begin{proof}
        Let $M\in(\C^{d\times d})^{\otimes 2}$ be an arbitrary matrix.
        The irreducible decomposition of $(\varPi_{\Yt{1}} |_G)^{\otimes 2}$ is given by $\left.\varPi_1\oplus\varPi_2\oplus \varPi_{\Yt{1,1}}\right|_G$ and the projections to their supporting subspaces are denoted by $P_1, P_2, P_{\Yt{1,1}}$, respectively.
        Twirling over $G$ gives
        \begin{equation}
            M^\prime \coloneqq 
            \mathcal{T}_{G}(M)
            =
            a_1 P_1 + a_2 P_2 + a_{\Yt{1,1}} P_{\Yt{1,1}},
        \end{equation}
        where $a_j \coloneqq \frac{\Tr(P_j M)}{\Tr(P_j)}$ for $j\in\left\{1,2,\Yt{1,1}\right\}$.
        The second twirling
        yields
        \begin{equation}
            \mathcal{T}_{G^{U}}
            \left(
                M^\prime
            \right)
            =
            b_1 P_1^{U^{\otimes 2}} + b_2 P_2^{U^{\otimes 2}} + b_{\Yt{1,1}} P_{\Yt{1,1}}^{U^{\otimes 2}},
        \end{equation}
        where $b_j \coloneqq \frac{\Tr(P_j^{U^{\otimes 2}} M^\prime)}{\Tr P_j}$ for $j\in\left\{1,2,\Yt{1,1}\right\}$.
        Since $P_{\Yt{1,1}}^{U^{\otimes 2}} = P_{\Yt{1,1}}$, their coefficients are the same $b_{\Yt{1,1}} = a_{\Yt{1,1}}$.
        Let $P_{\Yt{2}} = P_1 + P_2$ denote the projection to the symmetric subspace, for which we also know that $P_{\Yt{2}}^{U^{\otimes 2}} = P_{\Yt{2}}$.
        Substituting $P_2 = P_{\Yt{2}}-P_1$ into the expressions for the coefficients $b_1$ and $b_2$, we find that $b_1, b_2$ are completely determined by the coefficients $a_1, a_2$ and the value of $\Tr(P_1^{U^{\otimes 2}} P_1)$. 
        The condition from Eq.~\eqref{eq:desired_general} implies similar relations in the form of
        \begin{equation}\label{eq:overlap_symmetry}
            \Tr \left( P_i^{U^{\otimes 2}} P_j \right) = \frac{\Tr P_i \Tr P_{j}}{\Tr P_{\Yt{2}}} \qquad (i, j = 1, 2).
        \end{equation}
        This can be verified by
        \begin{align} \begin{split}
            \Tr \left( P_1^{U^{\otimes 2}} P_2 \right)
            = 
            \Tr \left( P_1^{U^{\otimes 2}} P_{\Yt{2}} \right)- \Tr \left( P_1^{U^{\otimes 2}} P_{1} \right)
            =
            \Tr P_1 - \frac{(\Tr P_1)^2}{\Tr P_{\Yt{2}}}
            =
            \frac{\Tr P_1 \Tr P_{2}}{\Tr P_{\Yt{2}}},
        \end{split}
        \end{align}
        where we used the fact that $\Tr(
            P_1^{U^{\otimes 2}} P_{\Yt{2}}
        ) = 
        \Tr(P_1) = d$.
        In a similar vein, we can show that
        \begin{equation}
            \Tr \left( P_2^{U^{\otimes 2}} P_2 \right)
            =
            \frac{(\Tr P_2)^2}{\Tr P_{\Yt{2}}}.
        \end{equation}
        The relation from Eq.~\eqref{eq:overlap_symmetry} enables us to write the final coefficients as
        \begin{subequations}
        \begin{align}
            b_1 &= \frac{\Tr(
                P_1^{U^{\otimes 2}} (a_1 P_1 + a_2 P_2)
            )}{\Tr P_1}
            =
            a_1\frac{\Tr P_1}{\Tr P_{\Yt{2}}}
            + a_2 \frac{\Tr P_2}{\Tr P_{\Yt{2}}},
            \\
            b_2 &= \frac{\Tr(
                P_2^{U^{\otimes 2}} (a_1 P_1 + a_2 P_2)
            )}{\Tr P_2}
            =
            a_1\frac{\Tr P_1}{\Tr P_{\Yt{2}}}
            + a_2 \frac{\Tr P_2}{\Tr P_{\Yt{2}}} \stackrel{!}{=} b_1.
        \end{align}
        \end{subequations}
        As a result, we know that $b_1 = b_2 \eqqcolon b_{\Yt{2}}$, meaning that the second twirling results in a quantity that can be equivalently obtained using a unitary $2$-design, i.e.,
        \begin{equation}
            \mathcal{T}_{G^{U} \cdot G} (M) = b_1 P_1 + b_2 P_2 + b_{\Yt{1,1}} P_{\Yt{1,1}}
            =
            b_{\Yt{2}} P_{\Yt{2}} + b_{\Yt{1,1}} P_{\Yt{1,1}}.
        \end{equation}
        Therefore, using \Cref{the_proposition}, the set $G^{U} \cdot G$ forms a unitary $2$-design.
    \end{proof}

    \subsection{Monomial reflection group}
        The monomial reflection group $G(m, 1, d)$ is the group of $d\times d$ matrices with a single non-zero entry in each row and column, where the only possible values for the non-zero entries are the $m$-th root of unity. 
        More generally, given $p$ such that $p\,|\,m$, $G(m,p,d)$ is the subgroup of $G(m,1,d)$ in which the product of the non-zero entries in each matrix is an $\frac{m}{p}$-th root of unity.
        However, for the purposes of this paper, the group $G(3, 1, d)$ is sufficient.
        The group has also been utilized for randomized benchmarking in Ref.~\cite{amaro2024randomised}.

        We know that $G(3, 1, d) \cong S_d \ltimes \mathbb{Z}_3^d$, since
        each element in $G(3, 1, d)$ can be completely characterized by
        \begin{itemize}
            \item an element in $S_d$ corresponding to the places of the nonzero entries in the $d\times d$ matrix
        through its natural representation, and
            \item an element in $\mathbb{Z}_3^d$ (corresponding to the $3$rd roots of unity in these entries).
        \end{itemize}
        The group product in $S_d \ltimes \mathbb{Z}_3^d$ is given by
        \begin{equation}\label{eq:monomial_semidirect}
            (\omega_1, \bm{\alpha}_1) \cdot (\omega_2, \bm{\alpha}_2) = (\omega_1 \omega_2, \bm{\alpha}_2 + \Gamma_{\text{nat}}(\omega_2) \bm{\alpha_1}) \qquad ( \omega_1, \omega_2 \in S_d, \; \bm{\alpha}_1, \bm{\alpha}_2 \in \mathbb{Z}^d_3),
        \end{equation}
        where $\Gamma_{\text{nat}}$ is the natural representation of the symmetric group $\S_d$ defined by $\Gamma_{\text{nat}}(\omega) \ket{j} \coloneqq \ket{\omega(j)}$.
        
        In the following, we will consider the \textit{natural representation} $\Gamma : G(3, 1, d) \cong  S_d \ltimes \mathbb{Z}_3^d \to \text{GL}(d, \C)$, defined as
        \begin{equation}
            \Gamma(\omega, \bm{\alpha}) \coloneqq \Gamma_{\text{nat}}(\omega) D(\bm{\alpha}) \qquad ( \omega \in S_d, \bm{\alpha} \in \mathbb{Z}^d_3),
        \end{equation}
        where, as before, $\Gamma_{\text{nat}}$ is the natural representation of the symmetric group $\S_d$  and \begin{equation}
            D(\bm{\alpha}) \coloneqq \overset{d-1}{\underset{j=0}{\bigoplus}}\;\zeta^{\alpha_j}
        \end{equation}
        is the representation of $\Z_3^d$ with $\zeta$ being the principal 3rd root of unity.
        In this representation, the product of two elements is
        \begin{align}
            \Gamma(\omega_1, \bm{\alpha}_1)
            \Gamma(\omega_2, \bm{\alpha}_2)
            = \Gamma_{\text{nat}}(\omega_1 \omega_2)
            D(\omega_2(\bm{\alpha}_1))
            D(\bm{\alpha}_2)
            = 
            \Gamma(\omega_1\omega_2, \bm{\alpha}_2 + \omega_2(\bm{\alpha}_1)),
        \end{align}
        analogously to Eq.~\eqref{eq:monomial_semidirect}.
        The natural representation of the symmetric group decomposes into two irreducible representations; one of them is the trivial representation acting on the subspace spanned by the vector $\sum_{j=0}^{d-1}\ket{j}$. However, this subspace is not left invariant by $D(\bm{\alpha})$, making $\Gamma$ irreducible.
        \begin{proposition}\label{prop:3irreps}
            Let $d \geq 2$. The $2$-fold tensor power of the representation $\Gamma$ has $3$ nontrivial invariant subspaces, which we will denote by $V_1, V_2, V_{\Yt{1,1}}$. Among these, $V_{\Yt{1,1}}$ is the antisymmetric subspace, $V_1$ is spanned by the vectors $\left\{ \ket{j}\otimes\ket{j} \right\}_{j=0}^{d-1}$ and $V_2$ is the complementary subspace of $V_1$ in the symmetric subspace $V_{\Yt{2}}$.
        \end{proposition}
        \begin{proof}
            Trivally, $V_1$ is invariant under the action of $\Gamma\otimes\Gamma$.
            Hence, if we prove that the irreducible decomposition of the $2$-fold tensor power of $\Gamma$ consists of $3$ inequivalent irreps, the rest of the statement follows trivially.
            To prove that $\Gamma \otimes \Gamma$ decomposes into three components, consider the character $\chi_{\Gamma\otimes\Gamma}$. We aim to show that $\left\langle \chi_{\Gamma\otimes\Gamma}, \chi_{\Gamma\otimes\Gamma} \right\rangle = 3$, i.e., $\Gamma\otimes\Gamma$ consists of three different irreducible components, resulting in $3$ different irreducible characters. The character of $\Gamma$ is
            \begin{align}
                \chi_{\Gamma}(\Gamma_{\text{nat}}(\omega)D(\bm{\alpha}))
                &=
                \sum_{j \in J(\omega)} \zeta^{\alpha_{j}},
            \end{align}
            where $J(\omega)$ is the set of indices representing the fixed points of $\omega$.
            With this equation, the inner product can be calculated as follows:
            \begin{equation}\begin{split}
                \left\langle \chi_{\Gamma^{\otimes 2}}, \chi_{\Gamma^{\otimes 2}}\right\rangle
                &=
                \left\langle \chi_{\Gamma}^2, \chi_{\Gamma}^2\right\rangle
                =
                \frac{1}{\abs{G(3,1,d)}}
                \sum_{g\in G(3,1,d)}\Tr(\Gamma(g))^2 \overline{\Tr(\Gamma(g))}^2
                \\&=
                \frac{1}{d! 3^d}
                \sum_{\omega \in S_d}
                \sum_{\bm{\alpha}\in \Z_3^d}
                \sum_{j,k \in J(\omega)}
                \sum_{\ell, m \in J(\omega)}
                \zeta^{\alpha_j + \alpha_k - \alpha_\ell - \alpha_m}
                \\&=\label{eq:triple_sums_for_inner_product}
                \frac{1}{3^d}
                \underset{\omega\in\S_d}{\E}\left[
                \sum_{j,k , \ell, m\in J(\omega)}
                \sum_{\bm{\alpha}\in \Z_3^d}
                \zeta^{\alpha_j + \alpha_k - \alpha_\ell - \alpha_m} \right].
            \end{split}\end{equation}
            It is clear that the innermost sum for fixed indices $j,k,\ell,m$ and fixed $\omega$ only depends on the equality of the indices $\left\{j,k\right\}=\left\{\ell,m\right\}$ as sets. We can write
            \begin{align}\label{eq:inner_sum_zeta_jklm}
                \sum_{\bm{\alpha}\in \Z_3^d}
                \zeta^{\alpha_j + \alpha_k - \alpha_\ell - \alpha_m}
                =
                \begin{cases}
                    3^d
                    &\text{if }
                    \left\{j,k\right\}=\left\{\ell,m\right\},
                    \\
                    0 &\text{otherwise.}
                \end{cases}
            \end{align}
            Let $f(\omega) \coloneqq \abs{J(\omega)}$.
            Hence, for fix $\omega$, we get the following:
            \begin{align}\label{eq:inner2sums}
                \sum_{j,k , \ell, m\in J(\omega)}
                \sum_{\bm{\alpha}\in \Z_3^d}
                \zeta^{\alpha_j + \alpha_k - \alpha_\ell - \alpha_m}
                &=
                \sum_{\substack{j,k , \ell, m\in J(\omega)\\\left\{j,k\right\}=\left\{\ell,m\right\}}}
                3^d
                =
                \left(
                    2f(\omega)^2 - f(\omega)
                \right)
                3^d
            \end{align}
            Proceeding with Eq.~\eqref{eq:triple_sums_for_inner_product}, in order to evaluate the average over $S_d$, we define $\iota_j(\omega)$ to be the characteristic function of $j$ being a fixed point of $\omega$.
            Then, the number of fixed points can be written as $f(\omega) = \sum_{j=0}^{d-1}\iota_j(\omega)$.
            The expectation values of $f(\omega)$ and $f^2(\omega)$ are the following:
            \begin{subequations}
            \begin{align}
                \underset{\omega\in\S_d}{\E} \iota_j(\omega) &=
                \frac{1}{d!}\sum_{\omega\in\S_d}
                \iota_j(\omega) = 
                \frac{(d-1)!}{d!}
                =
                \frac{1}{d},
                \\
                \underset{\omega\in\S_d}{\E} \iota_j(\omega)\iota_k(\omega) &= \frac{(d-2)!}{d!} =
                \frac{1}{d(d-1)}
                \qquad \text{for } j\neq k,
                \\
                \underset{\omega\in\S_d}{\E}f(\omega) &= 
                \sum_{j=0}^{d-1}\underset{\omega\in\S_d}{\E}\iota_j(\omega)
                = 1,
                \\
                \underset{\omega\in\S_d}{\E}f^2(\omega) &=
                \sum_{j,k=0}^{d-1}\underset{\omega\in\S_d}{\E}\iota_j(\omega)\iota_k(\omega)
                =
                \sum_{j}^{d-1}\underset{\omega\in\S_d}{\E}\iota_j(\omega)
                +
                \sum_{\substack{j,k=0\\j \neq k}}^{d-1}\underset{\omega\in\S_d}{\E}\iota_j(\omega)\iota_k(\omega) = 2.
            \end{align}
            \end{subequations}
            Finally, the inner product of the character follows from Eq.~\eqref{eq:triple_sums_for_inner_product} and Eq.~\eqref{eq:inner2sums}, and can be calculated as
            \begin{align}
                \left\langle \chi_{\Gamma^{\otimes 2}}, \chi_{\Gamma^{\otimes 2}}\right\rangle
                &=
                \frac{1}{d!}
                \sum_{\omega \in S_d}
                \left(
                    2f(\omega)^2 - f(\omega)
                \right)
                = 2\cdot 2 - 1 = 3.
            \end{align}
            The only possible way to get 3 as the sum of positive integer squares is by three $1$s, and hence, $\Gamma \otimes \Gamma$ decomposes into three inequivalent irreducible representations.
        \end{proof}
        In the following sections, we will denote the irreducible representations appearing in $\Gamma^{\otimes 2}$ by $\Gamma_1, \Gamma_2, \Gamma_{\Yt{1,1}}$, i.e.,
        \begin{equation}
            \Gamma^{\otimes 2} \cong \Gamma_1 \oplus \Gamma_2 \oplus \Gamma_{\Yt{1,1}}.
        \end{equation}

    \subsection{Constructing generalized unitary \texorpdfstring{$2$}{2}-design}
        As a consequence of Schur's lemma, twirling over a finite unitary subgroup yields a quantity that depends on the irreducible subspaces of the subgroup.
        In this subsection, two consecutive twirlings are used to get the same result as twirling by the entire unitary group, resulting in an exact unitary $2$-design by Proposition~\ref{the_proposition}.
        \begin{theorem}\label{thm:2design_in_arbitrary_dim}
            Let $d \geq 2$ and let $\mathcal{G}$ be the image of the representation $\Gamma$ of the monomial reflection group $G(3,1,d)$ and suppose that $Q$ is a unitary matrix with the following property:
            \begin{equation}\label{eq:desired_arbitrary_2design}
                \Tr\left(
                    P_1^{Q^{\otimes 2}} P_1
                \right)
                = \frac{2d}{d+1}.
            \end{equation}
            Then $\mathcal{G}^{Q} \cdot \mathcal{G}$ is a generalized unitary $2$-design, where $\mathcal{G}^{Q} = \left\{ Q g Q^\dagger \;|\; g \in \mathcal{G}\right\}$.
        \end{theorem}
        \begin{proof}
            Using that $\Tr P_1 = d$ and $\Tr P_{\Yt{2}} = d (d+1)/2$, the statement is a trivial consequence of \cref{lemma:2design_symmetric_subspace}.
        \end{proof}

        \Cref{thm:2design_in_arbitrary_dim} asserts the existence of a matrix $Q$ which fulfills the relation given by Eq.~\eqref{eq:desired_arbitrary_2design}.
        In the next part of this section, we construct a matrix that satisfies the required condition.
        The search for such a matrix follows the idea of defining the function
        \begin{align}
            f : \U(d) &\to \mathbb{R} \\
            U &\mapsto \Tr\left(
                P_1^{U^{\otimes 2}} P_1
            \right),
        \end{align}
        whose maximum value is 
        $d$, attained at the identity. The main idea is to find a matrix $B$ such that $f(B) \leq 2d/(d+1)$, and then to search for the optimal matrix $Q$ along a continuous path in $\U(d)$ between the identity and $B$.

        Let $F$ denote the matrix of the (normalized) discrete Fourier transform (DFT), with matrix elements given by
        \begin{equation}
            F_{jk} = 
            \frac{1}{\sqrt{d}}
            \exp\left(
                -\frac{2\pi i }{d}jk
            \right).
        \end{equation}
        This is a well-known transformation, but it is difficult to connect with a continuous path to the identity which is also analytically tractable.
        In order to acquire a simpler matrix, we can transform the DFT matrix as follows:
        \begin{equation}
            B \coloneqq \frac{1+i}{2} F + \frac{1-i}{2} F^\dagger,
        \end{equation}
        which is real valued and its eigenvalues are $\pm 1$.
        Importantly, since $B^2 = \mathbbm{1}$, we can give a parameterized path of unitaries that connects $B$ to the identity in $\U(d)$:
        \begin{align}\label{eq:Q_parametrized}
            Q(t) &\coloneqq \exp\left(
                i t B
            \right)
            =
            \cos(t) \mathbbm{1} + i \sin(t)B.
        \end{align}
        According to the next proposition, the path $Q(t)$ will contain at least one specific matrix $Q(t^\ast)$ that satisfies Eq.~\eqref{eq:desired_arbitrary_2design}.
        To keep this proposition compact, we define
        \begin{equation}\label{def:q_d}
            q(d) \coloneqq
            \begin{cases}
                0,& d\equiv 1,2,4,5 \pmod{8},\\[4pt]
                \sqrt{d},& d\equiv 7 \pmod{8},\\[4pt]
                -\sqrt{d},& d\equiv 3 \pmod{8},\\[4pt]
                \sqrt{2d},& d\equiv 0,6 \pmod{8}.
            \end{cases}
        \end{equation}

        \begin{proposition} \label{prop:trace_of_P1_rotatedP1}
            The overlap of $P_1$ and $P_1^{Q(t)^{\otimes 2}}$ is given by the following formula:
            \begin{align}\label{eq:trPPW_quadratic_in_sin2}
                \Tr\left(
                    P_1 P_1^{Q(t)^{\otimes 2}}
                \right)
                =
                a(d) \sin^4(t) + b(d) \sin^2(t) + d,
            \end{align}
            where
            \begin{subequations} 
            \begin{align}
                a(d) &\coloneqq \frac{2d-1}{2}
                -
                \frac{2 q(d)}{d}
                -
                \frac{
                    \operatorname{gcd}\left(
                        4, d
                    \right)
                }{2 d},
                \\
                b(d) &\coloneqq
                -2 (d-1)
                +
                \frac{2 q(d)}{d}.
            \end{align}
            \end{subequations}
            The desired value $\frac{2d}{d+1}$ can be reached by a specific $t$ given by 
            \begin{subequations}
            \begin{align}
                t &= \pm\arcsin(\sqrt{x^\ast}) + k \pi, \: k \in\mathbb{Z}, \\
                x^\ast &=\frac{-b(d)-\sqrt{\Delta(d)}}{2a(d)},
            \end{align}
            \end{subequations}
            with determinant $\Delta(d) = b(d)^2 - 4 a(d) c(d)$ with $c(d) \coloneqq d - \frac{2d}{d+1}$ the constant term in the $0$-reduced quadratic polynomial.
        \end{proposition}
        Proof in Appendix~\ref{app:rotating_P1}.

\section{Constructing unitary designs from orthogonal designs}\label{sec:constructing_unitary_design_from_orthogonal}
    \subsection{Orthogonal designs}
        The concept of a unitary $t$-design, which is based on the unitary group, may be naturally defined for other groups as well, see, e.g., Ref.~\cite{Chiribella_2013}.
        For our purposes, the orthogonal group $\Ort(d)$ will be of main concern.
        \begin{definition}[Orthogonal $t$-design] \label{def:orthogonal_t-design}
            A finite set $\mathcal{V} \subset \Ort(d)$ with weight function $w: \mathcal{V} \to [0, 1]$ is called a weighted orthogonal $t$-design if the following equation holds for any linear transformation $M$ on $(\C^{d})^{\otimes t}$:
            \begin{equation}
                 \sum_{V \in \mathcal{V} } w(V) \,
                 V^{\otimes t}M\left(V^{\otimes t}\right)^{\dagger}
                 =
                 \integral{O\in \Ort(d)}{} O^{\otimes t}M\left(O^{\otimes t}\right)^{\dagger} \mathrm{d} O,
            \end{equation}
            where the integral on the right-hand side is taken over all elements in $\Ort(d)$ with respect to the Haar measure. The number $t$ is called the degree of the design.
        \end{definition}
        
        In this section, let us use the notation
        \begin{align}
        \begin{split}
            \mathcal{T}_{\Ort(d)}^{(t)}(M) &\coloneqq \int_{\Ort(d)} O^{\otimes t} M (O^{\otimes t})^{\dagger} \mathrm{d}O,
            \\
            \mathcal{T}_{\U(d)}^{(t)}(M) &\coloneqq \int_{\U(d)} U^{\otimes t} M (U^{\otimes t})^{\dagger} \mathrm{d}U.
        \end{split}
        \end{align}
        Moreover, for a finite subset $\mathcal{V} \subset \U(d)$ with weight function $w$, we define 
        \begin{equation}
            \mathcal{F}_{\mathcal{V}}^{(t)}(M) \coloneqq \sum_{U \in \mathcal{V}} w(U) U^{\otimes t} M (U^{\otimes t})^{\dagger}.
        \end{equation}
        Here, we are considering approximate $t$-designs, since the main statement of this section holds in this general setting for $t \leq 3$. 
        In this paper, we are considering approximate designs using the relative error as follows~\cite{brandao2016local,schuster2025random}.
        \begin{definition}[$\epsilon$-approximate unitary $t$-design]
            \label{def:multiplicative_error_approximate}
            We call a finite subset $\mathcal{V} \subset \U(d)$ an $\epsilon$-approximate unitary $t$-design for some $\epsilon > 0$ when
            \begin{equation}
                \left(
                    1 - \epsilon
                \right)
                \mathcal{F}_{\mathcal{V}}^{(t)}
                \preceq
                \mathcal{T}_{\U(d)}^{(t)}
                \preceq
                \left(
                    1 + \epsilon
                \right)
                \mathcal{F}_{\mathcal{V}}^{(t)}
                ,
            \end{equation}
            where $\mathcal{A} \preceq \mathcal{A}^\prime$ means that the $\mathcal{A}^\prime - \mathcal{A}$ is a completely positive map.
        \end{definition}
        \begin{remark}
            The definition of $\epsilon$-approximate unitary design can be generalized to $\epsilon$-approximate orthogonal design naturally.
        \end{remark}

    \subsection{Representation theory of the orthogonal group}\label{representation_theory_of_the_orthogonal_groups}
        In this section, representation theoretic results are summerized, which are based on Appendix A of Ref.~\cite{solymos2024extendibilitybrauerstates}.
        We denote by $\varPi_{\Yt{1}}$ and $\Phi_{\Yt{1}}$ the defining representations of the unitary group $\U(d)$ and the orthogonal group $\Ort(d)$, respectively, with respect to the standard basis ${|j\rangle}_{j=0}^{d-1}$.
        The representation $\Phi_{\Yt{1}}$ can be embedded into $\varPi_{\Yt{1}}$.
        The irreducible decomposition of $\Phi_{\Yt{1}}^{\otimes 2}$ is the following~\cite{Hashagen_2018}:
        \begin{equation}
            \Phi_{\Yt{1}}^{\otimes 2} = \Phi_{\text{t}} \oplus \Phi_{\text{c}} \oplus \Phi_{\Yt{1,1}},
        \end{equation}
        where $\Phi_{\text{t}}$ is the trivial representation acting on the subspace $\mathcal{W}_{\text{t}}$ spanned by
        \begin{equation}\label{eq:psi_for_ortho_2designs}
            \left|\psi\right> =
            \frac{1}{\sqrt{d}}\sum_{j=0}^{d-1}
            \left| j \right> \otimes \left| j \right>.
        \end{equation}
        The representations $\Phi_{\text{t}},\Phi_{\text{c}}$ and $\Phi_{\Yt{1,1}}$ act on the subspaces $\mathcal{W}_{\text{t}}, \mathcal{W}_{\text{c}}$ and $\mathcal{W}_{\Yt{1,1}}$, respectively, where $\mathcal{W}_{\text{c}}$ is the complementary subspace of $\mathcal{W}_{\text{t}}$ in $V_{\Yt{2}}$.
        The projections to their subspaces are $P_{\text{t}} = \left| \psi \right>\!\left< \psi \right|, P_{\text{c}} = P_{\Yt{2}} - P_{\text{t}}$ and $P_{\Yt{1,1}}$, respectively.

        Let $\varPi_{\Yt{1}}$ denote the defining representation of the unitary group $\U(d)$.
        The irreducible decomposition of the third tensor power is
        \begin{equation}
            \varPi_{\Yt{1}}^{\otimes 3}
            \cong
            \varPi_{\Yt{3}}
            \oplus
            \varPi_{\Yt{2,1}}^{\oplus 2}
            \oplus
            \varPi_{\Yt{1,1,1}},
        \end{equation}
        where the underlying subspaces are denoted by $V_{\Yt{3}}, V_{\Yt{2,1}}, V_{\Yt{1,1,1}}$.
        It is known that the subspace $V_{\Yt{2,1}}$ supports two equivalent representations; hence, it can be written as $\widetilde{V}_{\Yt{2,1}} \otimes \C^2$ where the subspace $\widetilde{V}_{\Yt{2,1}}$ is the representation subspace and $\C^2$ is the multiplicity space.
        If we take the restriction of $\varPi_{\Yt{1}}^{\otimes 3}$ to the orthogonal group $\Ort(d)$, the prior two subspaces decompose into the following invariant subspaces:
        \begin{align}
            V_{\Yt{3}} &= \mathcal{W}_{\Yt{3}(0)} \oplus \mathcal{W}_{\Yt{3}(1)},
            \\
            \widetilde{V}_{\Yt{2,1}} &= \widetilde{\mathcal{W}}_{\Yt{2,1}(0)} \oplus \widetilde{\mathcal{W}}_{\Yt{2,1}(1)},
        \end{align}
        where the traceless components are labeled by $(0)$ and the traceful components are labeled with $(1)$. The trace components correspond to $d$-dimensional representations.
        The irreducible decomposition of $\Phi_{\Yt{1}}^{\otimes 3}$ is
        \begin{equation}
            \Phi_{\Yt{1}}^{\otimes 3}
            \cong
            \Phi_{\Yt{3}}
            \oplus
            \Phi_{\Yt{2,1}}^{\oplus 2}
            \oplus
            \Phi_{\Yt{1,1,1}}
            \oplus
            \Phi_{\Yt{1}}^{\oplus 3},
        \end{equation}
        where $\Phi_{\Yt{1}}^{\oplus 3}$ belongs to the subspace $\mathcal{W}_{\Yt{3}(1)} \oplus  \widetilde{\mathcal{W}}_{\Yt{2,1}(1)}^{\oplus 2}$.
        It is straightforward to see that the trace-component subspaces can be given by
        \begin{align}
            \mathcal{W}_{\Yt{3}(1)}
            &=
            \Span\left\{\sum_{j=0}^{d-1}\left(\ket{kjj} + \ket{jkj} + \ket{jjk}\right)\right\}_{k=0}^{d-1},
            \\
            \mathcal{W}_{\Yt{3}(1)} \oplus  \widetilde{\mathcal{W}}_{\Yt{2,1}(1)}^{\oplus 2}
            &=
            \Span\left\{
                \sum_{j=0}^{d-1}\ket{kjj},
                \sum_{j=0}^{d-1}\ket{jkj},
                \sum_{j=0}^{d-1}\ket{jjk}
            \right\}.
        \end{align}
        These facts are used to prove the statements of the following part of this paper.

    \subsection{Construction}
        In this section, a method is presented to construct a unitary $2$- or $3$-design from an orthogonal $2$- or $3$-design.
        
        For the $t=2$ and $t=3$ cases, we give a relation between orthogonal $t$-designs and unitary $t$-designs. For technical reasons, we start with the $t=2$ case:
        \begin{proposition} \label{prop:2twirling_ortho_uni}
            Suppose that $W$ is a unitary with the following property:
            \begin{equation}\label{eq:W_orthogonal_2design}
                \Tr\left(
                    P_{\text{t}}^{W^{\otimes 2}}
                    P_{\text{t}}
                \right)
                =
                \frac{2}{d(d+1)}.
            \end{equation}
            Twirling with the unitary group can be decomposed into two consecutive twirlings with the orthogonal group, i.e.,
            \begin{equation}\ytableausetup{aligntableaux=top}
                \mathcal{T}_{\U(d)}^{(2)} = \mathcal{T}_{\Ort(d)^W}^{(2)} \circ \mathcal{T}_{\Ort(d)}^{(2)},
            \end{equation}
            where $\Ort(d)^W \coloneqq \{W O W^\dagger : O \in \Ort(d) \}$.
        \end{proposition}
        
        \begin{proof}
            Using the notation from \Cref{representation_theory_of_the_orthogonal_groups}, the dimension of $\mathcal{W}_{\text{t}}$ is $1$ and the dimension of $V_{\Yt{2}}$ is $\frac{d(d+1)}{2}$.
            Based on the assumption Eq.~\eqref{eq:W_orthogonal_2design} and \Cref{lemma:2design_symmetric_subspace}, the proposition follows.
        \end{proof}
        To show the existence of unitary $W$ as assumed in Eq.~\eqref{eq:W_orthogonal_2design}, take the following diagonal matrix $W$:
        \begin{equation}\label{eq:W_alpha}
            W \ket{j} = \omega_{2d}^{j\alpha} \ket{j} \qquad j = 0, \dots, d-1,
        \end{equation}
        where $\omega_{2d} = \exp\left(\frac{2\pi i}{2d}\right)$ and $\alpha$ is the solution of the real equation
        \begin{equation}\label{eq:q_alpha}
            q(\alpha) = \left| \sum_{j=0}^{d-1} \omega_{2d}^{2\alpha j} \right|^2 = \frac{2d}{d+1}.
        \end{equation}
        Although $\alpha$ is not explicit, we can use the fact that $q(0) = d^2$, $q(1) = 0$ and the continuity of $q$ to efficiently find a solution of Eq.~\eqref{eq:q_alpha} for any dimension.
        To verify Eq.~\eqref{eq:W_orthogonal_2design}, with projection $P_{\text{t}} = \ketbra{\psi}{\psi},$ where $\ket{\psi}$ is defined as in Eq.~\eqref{eq:psi_for_ortho_2designs},
        we show that
        \begin{equation}\begin{split}
            \Tr\left(
                P_{\text{t}}^{W^{\otimes 2}}
                P_{\text{t}}
            \right)
            &= \Tr \left(
                W^{\otimes 2} \ketbra{\psi}{\psi} \left(W^\dagger\right)^{\otimes 2} \ketbra{\psi}{\psi}
            \right)^2
            =
            \left|
                \bra{\psi}W^{\otimes 2} \ket{\psi}
            \right|^2 \\
            &=
            \frac{1}{d^2}
            \left|
                \sum_{j,k=0}^{d-1}
                \bra{j} W
                \ket{k}^2
            \right|^2
            =
            \frac{1}{d^2}
            \left|
                \sum_{j,k=0}^{d-1}
                \delta_{jk}
                \omega_{2d}^{2k\alpha}
            \right|^2
            =
            \frac{1}{d^2}
            \left|
                \sum_{j=0}^{d-1}
                \omega_{2d}^{2\alpha j}
            \right|^2
            =
            \frac{2}{d(d+1)}.
        \end{split}\end{equation}

        Using the same unitary $W$, there is a similar connection between orthogonal and unitary $3$-designs as well:
        \begin{proposition}\label{prop:3twirling_ortho_uni}
            Let $d \geq 5$. For $t=3$, twirling with the unitary group can be decomposed into two consecutive twirlings with the orthogonal group, i.e.,
            \begin{equation}
                \mathcal{T}_{\U(d)}^{(3)}
                =
                \mathcal{T}_{\Ort(d)^W}^{(3)} \circ \mathcal{T}_{\Ort(d)}^{(3)},
            \end{equation}
            where $W$ is defined in Eq.~\eqref{eq:W_alpha}.
        \end{proposition}
        \begin{proof}
            In this proof, the notation as in \Cref{representation_theory_of_the_orthogonal_groups} is used.
            The statement of the proposition follows by showing that the result of the twirlings with the defining and rotated representations of the orthogonal group coincides with twirling with the unitary group.
            To establish this, we verify that on the subspaces $\mathcal{W}_{\Yt{3}(1)}$ and $\widetilde{\mathcal{W}}_{\Yt{2,1}(1)}^{\oplus 2}$, the maps $\mathcal{T}_{\U(d)}^{(3)}$ and $\mathcal{T}_{\Ort(d)^W}^{(3)} \circ \mathcal{T}_{\Ort(d)}^{(3)}$ act identically. This follows from the fact that both actions can be expressed solely in terms of the Young projectors $P{\Yt{3}}, P_{\Yt{2,1}}, P_{\Yt{1,1,1}}$ together with the projectors onto the subspaces $\mathcal{W}_{\Yt{3}(1)}$ and $\widetilde{\mathcal{W}}_{\Yt{2,1}(1)}^{\oplus 2}$.
            We will show the effect of the twirlings with the different orthogonal group representations on one general element of the matrix acting on $\mathcal{W}_{\Yt{3}(1)} \oplus  \widetilde{\mathcal{W}}_{\Yt{2,1}(1)}^{\oplus 2}$.
            Without loss of generality, let us consider the operator $\sum_{j,\ell=0}^{d-1}\ketbra{jjk}{m\ell\ell} \in \mathcal{W}_{\Yt{3}(1)} \oplus  \widetilde{\mathcal{W}}_{\Yt{2,1}(1)}^{\oplus 2}$,
            and 
            let $\mathbb{F}_{13}$ denote the flip between the first and third tensor components.
            It is known that the Bell state is invariant under the action of $(O\otimes O)$ for any $O \in \Ort(d)$, i.e.,
            \begin{equation}
                (O \otimes O) \sum_{j = 0}^{d-1}\ket{jj} 
                =
                \sum_{j = 0}^{d-1}\ket{jj}.
            \end{equation}
            Using the statement of Proposition~\ref{prop:2twirling_ortho_uni}, one can show that the two consecutive orthogonal $3$-twirlings give the same result as the unitary $3$-twirling.
            The first $3$-twirling gives
            \begin{equation}\begin{split}
                \mathcal{T}_{\Ort(d)}^{(3)}&
                \left(
                    \sum_{j,\ell=0}^{d-1}\ketbra{jjk}{m\ell\ell}
                \right)
                =
                \mathcal{T}_{\Ort(d)}^{(3)}
                \left(
                    \sum_{j,\ell=0}^{d-1}\ketbra{jjk}{\ell\ell m}
                    \mathbb{F}_{13}
                \right)
                \\&=
                \int_{O \in \Ort(d)}
                (\Id{} \otimes \Id{} \otimes O)
                \sum_{j,\ell=0}^{d-1}
                \left(
                    \ketbra{jjk}{\ell\ell m}
                \right)
                (\Id{} \otimes \Id{} \otimes O)^\dagger
                dO
                \mathbb{F}_{13}
                \\&=
                \left(
                    \sum_{j,\ell=0}^{d-1}
                    \left(
                        \ketbra{jj}{\ell\ell}
                    \right)
                    \otimes \frac{\delta_{km}}{d}\Id{}
                \right)
                \mathbb{F}_{13}.
            \end{split}\end{equation}
            A subsequent twirling with $\Ort(d)^W$ yields
            \begin{equation}\label{eq:3twirl_O3_example}\begin{split}
                \mathcal{T}_{\Ort(d)^W}^{(3)}
                &
                \left(
                    \left(
                        \sum_{j,\ell=0}^{d-1}
                        \left(
                            \ketbra{jj}{\ell\ell}
                        \right)
                        \otimes \frac{\delta_{km}}{d}\Id{}
                    \right)
                    \mathbb{F}_{13}
                \right) =
                \\&=
                \left(
                    \mathcal{T}_{\Ort(d)^W}^{(2)}\left(
                        \sum_{j,\ell=0}^{d-1}
                        \left(
                            \ketbra{jj}{\ell\ell}
                        \right)
                    \right)
                \otimes
                \frac{\delta_{km}}{d}\Id{}
                \right)
                \mathbb{F}_{13}
                \\&=
                \left(
                    \left(
                    \left(\mathcal{T}_{\Ort(d)^W}^{(2)}
                    \circ \mathcal{T}_{\Ort(d)}^{(2)}\right)
                        \sum_{j,\ell=0}^{d-1}
                        \left(
                            \ketbra{jj}{\ell\ell}
                        \right)
                    \right)
                \otimes
                \frac{\delta_{km}}{d}\Id{}
                \right)
                \mathbb{F}_{13}
                \\&=
                \left(
                    \mathcal{T}_{\U(d)}^{(2)}\left(
                        \sum_{j,\ell=0}^{d-1}
                        \left(
                            \ketbra{jj}{\ell\ell}
                        \right)
                    \right)
                \otimes
                \frac{\delta_{km}}{d}\Id{}
                \right)
                \mathbb{F}_{13}
                \\&=
                \left(\mathcal{T}_{\U(d)}^{(3)}
                \circ
                \mathcal{T}_{\Ort(d)}^{(3)}\right)
                \left(
                     \sum_{j,\ell=0}^{d-1}\ketbra{jjk}{m\ell\ell}
                \right)
                \\&=
                \mathcal{T}_{\U(d)}^{(3)}
                \left(
                     \sum_{j,\ell=0}^{d-1}\ketbra{jjk}{m\ell\ell}
                \right)
            \end{split}\end{equation}
            In the last equality, we have exchanged the $\U(d)$- and $\Ort(d)$-integrals and used the Haar invariance to obtain a single integral over $\U(d)$.
            Given that the action $\mathcal{T}_{\Ort(d)^W}^{(3)} \circ \mathcal{T}_{\Ort(d)}^{(3)}$ gives the same result as $\mathcal{T}_{\U(d)}^{(3)}$ in the subspaces $V_{\Yt{3}}, V_{\Yt{2,1}}, V_{\Yt{1,1,1}}$ and in the subspaces $\mathcal{W}_{\Yt{3}(1)}, \widetilde{\mathcal{W}}_{\Yt{2,1}(1)}^{\oplus 2}$, the statement is proven.
        \end{proof}
        
        \begin{lemma} \label{lemma:appr_ortho23_relative_error}
            Let $t \leq 3$.
            Consider an $\epsilon$-approximate orthogonal $t$-design  
            $\mathcal{V} \subset \Ort(d)$. Then $\mathcal{V}^{W}\cdot\mathcal{V}$ is a $(2\epsilon+\epsilon^2)$-approximate unitary $t$-design where $W \in \U(d)$ is a concrete unitary described in Proposition~\ref{prop:2twirling_ortho_uni}.
        \end{lemma}
        \begin{proof}
            Note that throughout the proof we restrict to $t \leq 3$, although we keep the notation simply as $t$ for readability.
            Since $\mathcal{V}$ is an $\epsilon$-approximate orthogonal $t$-design, by definition we know that
            \begin{equation} \label{eq:assumption_ortho2design}
                (1 - \epsilon) \mathcal{F}_{\mathcal{V}}^{(t)}
                \preceq
                \mathcal{T}_{\Ort(d)}^{(t)}
                \preceq
                (1 + \epsilon) \mathcal{F}_{\mathcal{V}}^{(t)} \,\text{.}
            \end{equation}
            Our goal is to establish
            \begin{equation}\label{eq:approximate_ortho_23design_goal}
                (1 - 2\epsilon - \epsilon^2) \mathcal{F}_{\mathcal{V}^W \cdot \mathcal{V}}^{(t)}
                \preceq
                \mathcal{T}_{\U(d)}^{(t)}
                \preceq
                (1 + 2\epsilon + \epsilon^2) \mathcal{F}_{\mathcal{V}^W \cdot \mathcal{V}}^{(t)}
                \,\text{.}
            \end{equation}
            By \Cref{prop:2twirling_ortho_uni} and \Cref{prop:3twirling_ortho_uni}, the unitary $t$-twirl decomposes as $\mathcal{T}_{\U(d)}^{(t)} = \mathcal{T}_{\Ort(d)^W}^{(t)} \circ \mathcal{T}_{\Ort(d)}^{(t)}$.
            Moreover, $\mathcal{F}_{\mathcal{V}^W \cdot \mathcal{V}}^{(t)} = \mathcal{F}_{\mathcal{V}^W}^{(t)} \circ \mathcal{F}_{\mathcal{V}}^{(t)}$.
            The inequalities in Eq.~\eqref{eq:assumption_ortho2design} also hold for $\mathcal{V}^W$ because conjugation by $W$ preserves complete positivity and operator order.
            Since operator inequalities are preserved under left- and right-composition with completely positive maps, we obtain
            \begin{subequations}\begin{align}
                \mathcal{T}_{\U(d)}^{(t)}
                &=
                \mathcal{T}_{\Ort(d)^W}^{(t)}
                \circ
                \mathcal{T}_{\Ort(d)}^{(t)}
                \succeq
                (1 - \epsilon)
                \mathcal{T}_{\Ort(d)^W}^{(t)}
                \circ
                \mathcal{F}_{\mathcal{V}}^{(t)}
                \succeq
                (1 - \epsilon)^2
                \mathcal{F}_{\mathcal{V}^{W}}^{(t)}
                \circ
                \mathcal{F}_{\mathcal{V}}^{(t)}
                \text{,}
                \\
                \mathcal{T}_{\U(d)}^{(t)}
                &=
                \mathcal{T}_{\Ort(d)^W}^{(t)}
                \circ
                \mathcal{T}_{\Ort(d)}^{(t)}
                \preceq
                (1 + \epsilon)
                \mathcal{T}_{\Ort(d)^W}^{(t)}
                \circ
                \mathcal{F}_{\mathcal{V}}^{(t)}
                \preceq
                (1 + \epsilon)^2
                \mathcal{F}_{\mathcal{V}^{W}}^{(t)}
                \circ
                \mathcal{F}_{\mathcal{V}}^{(t)}
                \text{.}
            \end{align}\end{subequations}
            The above inequalities give a strictly stronger lower bound than the one stated on the LHS of Eq.~\eqref{eq:approximate_ortho_23design_goal}, while the corresponding upper bound coincides exactly with the RHS of Eq.~\eqref{eq:approximate_ortho_23design_goal}.
        \end{proof}
        
        As a special case, an exact orthogonal $2$-design or $3$-design yields an exact unitary $2$-design or $3$-design, respectively.
        \begin{theorem} \label{thm:23design_from_orthogonal23}
            Let $t \leq 3$.
            Consider an orthogonal $t$-design $\mathcal{V} \subset \Ort(d)$. Then there exists a unitary $W \in \U(d)$ such that the set $\mathcal{V}^{W} \cdot \mathcal{V}$ forms a unitary generalized group $t$-design.
        \end{theorem}
        \begin{proof}
            Trivial based on \Cref{lemma:appr_ortho23_relative_error} by taking $\epsilon \to 0$.
        \end{proof}

        To showcase the use of Theorem \ref{thm:23design_from_orthogonal23}, we present a number of examples.
        With GAP systems, it can be shown that several groups have orthogonal irreducible representations whose image constitutes a $t$-design with $t=2$ or $t=3$(from several possible irreps, just one is selected).
        From these irreducible representations, exact unitary $2$- or $3$-designs can be constructed using Theorem~\ref{thm:23design_from_orthogonal23}.
        \begin{table}[h]
            \centering
            \begin{tabular}{lccc}
            \hline
            \textbf{Name of group in GAP} & \textbf{Index of irrep} & \textbf{Dimension} & t \\
            \hline
            \texttt{(3\^{}2:2xG2(3)).2} & 31 & 14 & 3 \\
            \texttt{(2x2\^{}(1+8)\_+).O8+(2)} & 156 & 16 & 3 \\
            \texttt{Co2} & 2 & 23 & 3 \\
            \texttt{2.Co1} & 102 & 24 & 3 \\
            \texttt{Isoclinic(2.F4(2).2)} & 92 & 52 & 3 \\
            \hdashline
            \texttt{2.U6(2)} & 2 & 22 & 2 \\
            \texttt{S6(5)} & 2 & 63 & 2 \\
            \texttt{Fi22} & 2 & 78 & 2 \\
            \texttt{HN} & 2 & 133 & 2 \\
            \texttt{Th} & 2 & 248 & 2 \\
            \hline
            \end{tabular}
            \caption{Data extracted from GAP (group name, irrep index, dimension, and design degree $t$) for irreducible representations whose image yields orthogonal $t$-designs. (Wherever several possible irreps are available for our use, just one is selected.)
            These representations provide explicit examples that can be used in \Cref{thm:23design_from_orthogonal23} to construct exact unitary $2$- and $3$-designs.
            All computations were carried out using GAP 4.14.0~\cite{GAP4}.
            }
            \label{table:orthogonal_3designs_examples_GAP}
        \end{table}
        Since the Clifford group on four qubits forms a unitary 3-design in dimension 16, the example in the paper merely illustrates a direct application of the theorem.

\section{Summary and Outlook}
    The current paper establishes the concept of a generalized group $t$-design and provides several construction methods.
    One such construction is enabled by \Cref{thm:creating_t_design_from_groups}, establishing a procedure for constructing a generalized group $t$-design using finite groups whose representations admit an easily verifiable property.
    Moreover, we mention some examples of novel $3$- and $4$-design constructions using this procedure. Another notable construction is given in \Cref{thm:2design_in_arbitrary_dim}, which allows one to construct an explicit generalized group $2$-design in arbitrary dimensions using the natural representation of the monomial reflection group. We have also described a method for converting orthogonal $2$- or $3$-designs into unitary designs in \Cref{thm:23design_from_orthogonal23}, together with several examples.

    There is an abundance of possible research directions based on group representations and the ideas presented here.
    As a start, we plan to carry out a thorough search through the finite groups using the GAP system to identify cases where $5$- or higher-degree-designs could be constructed using \Cref{thm:creating_t_design_from_groups}.
    On the more ambitious side, one could envision obtaining a construction for $3$-designs for general dimensions in a similar way to \Cref{thm:2design_in_arbitrary_dim}.
    A further research direction is to examine the consequences of our $2$- and higher-degree-design constructions in applications such as quantum process tomography, randomized benchmarking, and (thrifty) shadow estimation.

\section*{Acknowledgement}
    We thank David Amaro Alcalá for fruitful discussions.
    This research was supported by the Ministry of Culture and Innovation and the National Research, Development and Innovation Office (NKFIH) through the Quantum Information National Laboratory of Hungary (Grant No. 2022-2.1.1-NL-2022-00004), the AI4QT project Nr. 2020-1.2.3-EUREKA-2022-00029 and the Grant No. FK135220.
    The authors were also supported by the QuantERA project HQCC-101017733 and ZZ was partially supported by the Horizon Europe programme HORIZON-CL4-2022-QUANTUM-01-SGA via the project 101113946 OpenSuperQPlus100.

\newpage

\renewcommand{\thesection}{Appendix \Alph{section}}

\appendix

\section{Rotation of the \texorpdfstring{$d$}{d}-dimensional irreducible representation of the monomial reflection group \texorpdfstring{$G(3,1,d)$}{G(3,1,d)}} \label{app:rotating_P1}
    In this section, the proof of Proposition~\ref{prop:trace_of_P1_rotatedP1} is presented in detail.
    The first step describes the overlap of the projection $P_1$ and its rotation $P_1^{Q^{\otimes 2}}$ in terms of the matrix elements of $Q$:
    \begin{proposition}
        Let $P_1$ be the projection onto the subspace spanned by $\{\ket{j}\otimes\ket{j}\}_{j=0}^{d-1}$.
        The trace of the product of $P_1$ and $P_1$ rotated by $Q$ is given by the equation:
        \begin{equation}\label{eq:tr_P1_rotatedP1_general}
            \Tr(P_1^{Q^{\otimes 2}} P_1) =
            \sum_{j,k = 0}^{d-1}
            \left|\left\langle
                j
                \middle|
                Q
                \middle|
                k
            \right\rangle\right|^4
        \end{equation}
    \end{proposition}
    \begin{proof}
        By direct calculation, we get
        \begin{align}\begin{split}
            \Tr\left(
                P_1^{Q^{\otimes 2}} P_1
            \right)
            &=
                \sum_{\ell=0}^{d-1}
                \sum_{j = 0}^{d-1}
                \sum_{k = 0}^{d-1}
                \overbrace{\braket{\ell\ell}{jj}}^{\delta_{\ell j}}
                \bra{jj}
                Q^{\otimes 2} 
                \ketbra{kk}{kk}
                Q^{\dagger\otimes 2}
                \ket{\ell \ell}
            \\&=
                \sum_{j = 0}^{d-1}
                \sum_{k = 0}^{d-1}
                \bra{jj}
                Q^{\otimes 2} 
                \ketbra{kk}{kk}
                Q^{\dagger\otimes 2}
                \ket{jj}
            \\&=
                \sum_{j,k = 0}^{d-1}
                \left|\left\langle
                    j
                    \middle|
                    Q
                    \middle|
                    k
                \right\rangle\right|^4.
        \end{split}\end{align}
    \end{proof}
    To proceed, let us introduce the notation
    \begin{equation}
        \tau_{jk} \coloneqq \frac{2\pi j k }{d},
    \end{equation}
    for brevity.
    To obtain explicit formulas when we use Eq.~\eqref{eq:Q_parametrized} in Eq.~\eqref{eq:tr_P1_rotatedP1_general}, two quadratic Gauss sums have to be considered:
    \begin{proposition}
        Using the previous notation, the following identities hold:
        \begin{align}
            \label{eq:sin-sum}
            \sum_{j=0}^{d-1}\sin(2\tau_{jj})
            &=
            \sum_{j=0}^{d-1}
            \sin\!\left(\frac{4\pi j^2}{d}\right)
            =
            q(d)
            \coloneqq 
            \begin{cases}
                0,& d\equiv 1,2,4,5 \pmod{8},\\[4pt]
                \sqrt{d},& d\equiv 7 \pmod{8},\\[4pt]
                -\sqrt{d},& d\equiv 3 \pmod{8},\\[4pt]
                \sqrt{2d},& d\equiv 0,6 \pmod{8},
            \end{cases}
            \\ \label{eq:cos-sum}
            \sum_{j,k=0}^{d-1}\cos(4\tau_{jk})
            &=\sum_{j,k=0}^{d-1}\cos\!\left(\frac{8\pi j k}{d}\right)
            = d\cdot\gcd(4,d),
        \end{align}
        where the function $q(d)$ is defined by Eq.~\eqref{def:q_d}.
    \end{proposition}
    
    \begin{proof}
        We prove the two identities separately.
        
        \textbf{The sine diagonal-sum (Eq.~\eqref{eq:sin-sum}).}
        Introducing the quadratic Gauss sum for any integer $a$ and positive integer $d$:
        \begin{equation}
            g(a,d) \coloneqq
            \sum_{j=0}^{d-1} e^{2\pi i \frac{aj^2}{d}}.
        \end{equation}
        Then the function $q$ can be written as
        \begin{equation}
            q(d) = \Im\left(
                g(2, d)
            \right).
        \end{equation}
        The following statements regarding Gauss sums follow from results in Ref.~\cite{berndt1998gauss}, in particular Equations~\labelcref{eq:im_g_1_d,eq:g_2_d,eq:im_g_2_d}.
        The quadratic Gauss sum can be evaluated for the case of $a=1$:
        \begin{equation}\label{eq:im_g_1_d}
            \Im \left(g(1, d)\right) = \begin{cases}
                0 & \text{if } a \equiv 1,2 \pmod{4}, \\
                \sqrt{d} & \text{if } a \equiv 0,3 \pmod{4}.
            \end{cases}
        \end{equation}
        For odd $d$, the Gauss sum can be written as
        \begin{equation}\label{eq:g_2_d}
            g(2,d) = \epsilon_{d} \sqrt{d} (-1)^{\frac{d^2-1}{8}},
        \end{equation}
        where $\epsilon_d = 1$ if $d \equiv 1 \pmod{4}$ and $\epsilon_d = i$ if $m \equiv 3 \pmod{4}$.
        For odd $d$ this results in the following
        \begin{equation}\label{eq:im_g_2_d}
            q(d) =
            \Im \left(g(2, d) \right) =
            \begin{cases}
                0 & \text{if } d \equiv 1, 5 \pmod{8} ,\\
                \sqrt{d} & \text{if } d \equiv 7 \pmod{8} ,\\
                -\sqrt{d} & \text{if } d \equiv 3 \pmod{8},
            \end{cases}
            \qquad (\text{for odd $d$}).
        \end{equation}
        For even $d$, the greatest common divisor is $\gcd(2, d) = 2$. 
        Using $\gcd$ reduction for even $d$ we get
        \begin{equation}
            q(d) = \Im \left(
                g(2, d)
            \right)
            = 2 \Im \left(
                g(1, d/2)
            \right).            
        \end{equation}
        Using Eq.~\eqref{eq:im_g_1_d}, for even $d$ we acquire
        \begin{equation}
            q(d) = \begin{cases}
                0 & \text{if } d \equiv 2, 4 \pmod{8}, \\
                \sqrt{2d} & \text{if } d \equiv 0, 6 \pmod{8},
            \end{cases}
            \qquad (\text{for even $d$}).
        \end{equation}
        Putting together the even and odd cases, we get Eq.~\eqref{eq:sin-sum}.

        \smallskip

        \textbf{The cosine double-sum (Eq.~\eqref{eq:cos-sum}).}
        The double-sum can be written as the real part of an exponential sum:
        \begin{equation}
            \sum_{j,k=0}^{d-1}\cos\!\left(\frac{8\pi j k}{d}\right)
            =\Re\!\left(\sum_{j,k=0}^{d-1} e^{2\pi i \frac{4 j k}{d}}\right)
            =\Re\!\left(\sum_{j=0}^{d-1}\sum_{k=0}^{d-1} e^{2\pi i \frac{(4j)k}{d}}\right).
        \end{equation}
        For fixed $j$, the inner sum $\sum_{k=0}^{d-1} e^{2\pi i \frac{(4j)k}{d}} = d$ if $4j \equiv 0\pmod{d}$ and equals $0$ otherwise (geometric series). Thus the whole double sum equals
        \begin{equation}
            d\cdot\#\{j\in\{0,\dots,d-1\}: d\mid 4j\}.
        \end{equation}
        The congruence $d\mid 4j$ has exactly $\gcd(4,d)$ solutions modulo $d$ (indeed $4j\equiv 0\pmod{d}$ if and only if $j$ is a multiple of $d/\gcd(4,d)$, and there are $\gcd(4,d)$ residue classes of that form). This yields Eq.~\eqref{eq:cos-sum}.
    \end{proof}

    The preceding two propositions serve as a preliminary introduction to this Appendix. Their purpose is to prepare the ground for the proof of Proposition~\ref{prop:trace_of_P1_rotatedP1}, which is the main objective of this section. We now present the proof of Proposition~\ref{prop:trace_of_P1_rotatedP1}:
    \begin{proof}[Proof of Proposition~\ref{prop:trace_of_P1_rotatedP1}]
        The sum can be split into the diagonal values and the off-diagonal ones.
        \begin{align}
            \sum_{j,k=0}^{d-1}
            \left|\left\langle
                j
                \middle|
                Q(t)
                \middle|
                k
            \right\rangle\right|^4
            &=
            \sum_{j,k=0}^{d-1}
                (Q(t)_{jk})^4
            \\&=
            \sum_{j}^{d-1}
                (Q(t)_{jj})^4
            +
            \sum_{\substack{j,k=0\\j\neq k}}^{d-1}
                (Q(t)_{jk})^4.
        \end{align}
        The elements of the matrix $Q(t)$ are
        \begin{equation}
            Q(t)_{jk}
            =
            \delta_{jk}
            \cos(t)
            +
            \frac{i}{\sqrt{d}}
            \sin(t)
            \left(
                \cos(\tau_{jk})+\sin(\tau_{jk})
            \right)
            ,
        \end{equation}
        where $\tau_{jk} = 2\pi jk/d$.
        The sum of all elements can be split into two different sums corresponding to the diagonal elements and the others.
        \begin{align}\label{eq:sum_qt_jk}
            \sum_{j,k=0}^{d-1}
            \left|
                Q(t)_{jk}
            \right|^4
            &=
            \sum_{j=0}^{d-1}
            \left|
                Q(t)_{jj}
            \right|^4
            +
            \sum_{\substack{j,k=0\\j \neq k}}^{d-1}
            \left|
                Q(t)_{jk}
            \right|^4.
        \end{align}
        First, the first (diagonal) sum is expressed as
        \begin{align}
            \sum_{j=0}^{d-1}
            \left|
                Q(t)_{jj}
            \right|^4
            &=
            \sum_{j=0}^{d-1}
            \left(
                \cos(t)^2
                +
                \frac{1}{d}
                \sin(t)^2
                \left(
                    \cos(\tau_{jj})+\sin(\tau_{jj})
                \right)^2
            \right)^2
            \\&=
            \sum_{j=0}^{d-1}
            \left(
                1
                +
                \sin(t)^2
                \overbrace{
                \left(
                    \frac{1}{d}
                    \left(
                        \cos(\tau_{jj})+\sin(\tau_{jj})
                    \right)^2
                    -1
                \right)
                }^{\Theta_j}
            \right)^2
            \\&=
            d
            +
            2
            \sin(t)^2
            \sum_{j=0}^{d-1}
            \Theta_j
            +
            \sin(t)^4
            \sum_{j=0}^{d-1}
            \Theta_j^2.
        \end{align}
        The second sum can be expressed as
        \begin{align}
            \sum_{\substack{j,k=0\\j \neq k}}^{d-1}
            \left|
                Q(t)_{jk}
            \right|^4
            &=
            \frac{1}{d^2}
            \sin(t)^4
            \sum_{\substack{j,k=0\\j \neq k}}^{d-1}
            \left|
                \left(
                    \cos(\tau_{jk})+\sin(\tau_{jk})
                \right)
            \right|^4.
        \end{align}
        The Eq~\eqref{eq:sum_qt_jk} is a $2$-degree polynomial in $\sin(t)^2$ with coefficients $a(d), b(d)$ and $d$.
        The following equations are used in the calculations:
        \begin{subequations}
        \begin{align}
            (\cos(x) + \sin(x))^2 &=
            1 + \sin(2x),
            \\
            (\cos(x) + \sin(x))^4 &=
            \frac{3}{2}
            +
            2 \sin(2x)
            -
            \frac{1}{2}
            \cos(4x).
        \end{align}
        \end{subequations}
        The linear coefficient is
        \begin{equation}\begin{split}
            b(d) &=
            2
            \frac{1}{d}
            \sum_{j=0}^{d-1}
            \left(
                \cos(\tau_{jj})+\sin(\tau_{jj})
            \right)^2
            -
            2
            \sum_{j=0}^{d-1}
                1
            \\&=
            \frac{2}{d}
            \sum_{j=0}^{d-1}
            \left(
                1 + \sin(2\tau_{jj})
            \right)
            -
            2
            d
            \\&=
            2(1-d)
            +
            \frac{2}{d}
            \sum_{j=0}^{d-1}
            \sin(2\tau_{jj}).
        \end{split}\end{equation}
        
        The coefficient of $\sin(t)^4$ is
        \begin{equation}\begin{split}
            a(d) &=
            \frac{1}{d^2}
            \sum_{j=0}^{d-1}
            \left(
                \left(
                    \cos(\tau_{jj})
                    +
                    \sin(\tau_{jj})
                \right)^2
                -d
            \right)^2
            +
            \frac{1}{d^2}
            \sum_{\substack{j,k=0\\j \neq k}}^{d-1}
            \left(
                    \cos(\tau_{jk})+\sin(\tau_{jk})
            \right)^4
            \\&=
            \frac{1}{d^2}
            \sum_{\substack{j,k=0}}^{d-1}
            \left(
                    \cos(\tau_{jk})+\sin(\tau_{jk})
            \right)^4
            -
            \frac{2}{d}
            \sum_{j=0}^{d-1}
            \left(
                \cos(\tau_{jj})
                +
                \sin(\tau_{jj})
            \right)^2
            +
            d
            \\&=
            \frac{1}{d^2}
            \sum_{ \substack{j,k=0}}^{d-1}
            \left(
                \frac{3}{2}
                +
                2
                \sin(2\tau_{jk})
                -
                \frac{1}{2}
                \cos(4\tau_{jk})
            \right)
            -
            \frac{2}{d}
            \sum_{j=0}^{d-1}
            \left(
                1 + \sin(2\tau_{jj})
            \right)
            +
            d
            \\&=
            \frac{3}{2}
            +
            \frac{2}{d^2}
            \sum_{j,k=0}^{d-1}
            \sin(2\tau_{jk})
            -
            \frac{1}{2}
            \frac{1}{d^2}
            \sum_{j,k=0}^{d-1}
            \cos(4\tau_{jk})
            -
            2
            -
            \frac{2}{d}
            \sum_{j=0}^{d-1}
            \sin(2\tau_{jj})
            +
            d
            \\&=
            \frac{2d-1}{2}
            +
            \frac{2}{d^2}
            \overbrace{
                \sum_{j,k=0}^{d-1}
                \sin(2\tau_{jk})
            }^{0}
            -
            \frac{1}{2d^2}
            \sum_{j,k=0}^{d-1}
            \cos(4\tau_{jk})
            -
            \frac{2}{d}
            \sum_{j=0}^{d-1}
            \sin(2\tau_{jj}).
        \end{split}\end{equation}
    \end{proof}

    \begin{proposition}
        Let $d\ge 2$ and let $a(d), b(d)$ be the coefficients in the displayed
        equation from Proposition~\ref{prop:trace_of_P1_rotatedP1}:
        \begin{equation}
            \frac{2d}{d+1}=a(d)\sin^4 t + b(d)\sin^2 t + d,
        \end{equation}
        so that with $x=\sin^2(t)$ we obtain the quadratic equation
        \begin{equation}\label{eq:quad}
            f(x) = a(d)x^2+b(d)x+c(d)=0,
        \end{equation}
        where $c(d) \coloneqq d-\frac{2d}{d+1}$.
        Let $\Delta(d) \coloneqq b(d)^2-4 \: a(d) \: c(d)$.
        Then $\Delta(d) \ge 0$ and exactly one root of Eq.~\eqref{eq:quad} lies in the interval $(0,1)$.
        That root can be written in the form
        \begin{equation}\label{eq:root_xast}
            x^\ast=\frac{-b(d)-\sqrt{\Delta(d)}}{2a(d)},
        \end{equation}
        and it satisfies $0<x^\ast<1$; conversely, any admissible solution for $\sin^2(t)$ is $x^\ast$.
    \end{proposition}

    \begin{proof}
        We will prove three facts: (i) $f(0) = c(d)>0$, (ii) $f(1)=a(d)+b(d)+c(d)<0$, and (iii)
        $\Delta\ge0$. The facts (i)--(iii) together imply there is exactly one real root
        in $(0,1)$ and it is given by Eq.~\eqref{eq:root_xast}.
        
        (i) \emph{Positivity at $0$.}
        Since the constant term is $d$, we have
        \begin{equation}
        c(d) = d - \frac{2d}{d+1} = \frac{d(d-1)}{d+1} > 0\qquad(\text{for } d\ge2),
        \end{equation}
        hence $f(0)=c(d)>0$.
        
        (ii) \emph{Negativity at $1$.}
        Compute $f(1)=a(d)+b(d)+c(d)$. Using the given explicit forms of $a(d)$ and $b(d)$ we get that
        \begin{equation}
        a(d)+b(d) = -\frac{2d-3}{2} - d\cdot\gcd(4,d) + \frac{2}{d}q(d) - \frac{2}{d}q(d)
        = -\frac{2d-3}{2} - d\cdot\gcd(4,d),
        \end{equation}
        Hence, $f(1)$ is given by
        \begin{equation}\begin{split}
        f(1)&=a(d)+b(d)+c(d) = -\frac{2d-3}{2} - d\cdot\gcd(4,d) + \frac{d(d-1)}{d+1}.
        \\&=
        \frac{
            3-d -2d(d+1)\cdot \gcd(4,d)
        }{2(d+1)} < 0\qquad \text{if } d \ge 2,
        \end{split}\end{equation}
        since the numerator is less than $0$ for any $d \ge 1$.
        As $f(x)$ is a quadratic polynomial, it can only contain one or two roots in the open interval $(0, 1)$.
        
        \emph{(iii) Discriminant and uniqueness.}
        We have already shown that there is at least one real root in $(0,1)$; therefore
        $\Delta(d) > 0$ and there are real roots. Since $c(d)>0$, the product of the two
        roots is $\frac{c(d)}{a(d)}$. Two cases occur depending on the sign of $a(d)$:
        
        \begin{itemize}
        \item If $a(d)>0$, the parabola opens upwards. Then, the conditions $f(0)>0$ and $f(1)<0$
        imply that the smaller root (the one with the minus sign in the standard formula)
        lies in the open interval $(0,1)$. Concretely, the root in $(0, 1)$ has to be the one given by Eq.~\eqref{eq:root_xast}.
        
        \item If $a(d) < 0$, the parabola opens downward. In that case, the greater root has to be in the interval $(0, 1)$: the expression with the minus sign represents the greater root since $a(d) < 0$.
        Thus, the in-interval root can be written uniformly as
        \begin{equation}
            x^\ast=\frac{-b(d)-\sqrt{\Delta(d)}}{2a(d)},
        \end{equation}
        which coincides with Eq.~\eqref{eq:root_xast}.
        \end{itemize}
        
        Therefore, $t=\pm\arcsin(\sqrt{x^\ast})+\pi k$ with some $k\in \mathbb{Z}$.
    \end{proof}

\section{Recap of representation theory}\label{app:representation_theory}

The representation theory of finite and compact groups is a vast topic, and the interested reader is directed to \cite{FH,kirillov1976elements, goodman2009symmetry, christandl2006structurebipartitequantumstates} for more thorough introductions. In this section, only the basic notions are presented.

In quantum theory, we are usually interested in operations that can be described by unitary matrices. Thus, every representation we encounter in this paper is a unitary representation, i.e., a representation $\varPi:G\to \U(d)$.
The dimension of the representation is the dimension of the underlying vector space $\mathbb{C}^d$ on which these matrices act, i.e., $\dim \varPi=d$.

Two unitary representations are considered the same (i.e., unitarily equivalent) if there is a unitary basis transformation that links them. In general, a group can have many different representations. Sometimes, one talks about \textit{natural, standard} or \textit{defining} representations when the representation is based on some native properties of the group itself. Thus, in the case of matrix groups such as $\U(d)$, the defining representation is the representation when every unitary matrix in $\U(d)$ is represented by itself.

To get a handle on representations, we usually turn to representations that are, in some sense, indivisible into smaller ones:
\begin{definition}[Irreducible representation]
    A representation $\varPi:G\to \U(d)$ is called \textit{irreducible} (or irrep for short) if there is no non-trivial subspace in $\mathbb{C}^d$ such that every matrix $\varPi(g)$ leaves it invariant. That is, a non-trivial subspace $V\subset \mathbb{C}^{d}$ such that $\varPi(g)v\in V$ for all $g\in G$ and $v\in V$.
\end{definition}

A nice property of unitary representations is the following:
\begin{proposition}
Every unitary representation of any group $G$ is completely reducible. That is, every unitary representation $\varPi:G\to \U(d)$ can be written as the direct sum of irreducible representations. In such a direct sum decomposition we call the isotypic component of an irrep $\varPi_\gamma$ the sum of all subrepresentations isomorphic to that given irrep.
\end{proposition}

For example, if a unitary representation $\varPi:G\to\U(d)$ splits up into four irreps, out of which two are isomorphic, then we have three isotypic components:
\begin{equation}
\varPi\cong \varPi_1 \oplus \left( \varPi_2 \oplus \varPi_2\right) \oplus \varPi_3\text{.}
\end{equation}
This means, that the underlying vector space $\mathbb{C}^d$ splits into direct summands under the action of $\varPi$ as:
\begin{equation}
\mathbb{C}^d \cong \mathbb{C}^{d_1} \oplus \left(\mathbb{C}^{d_2} \oplus \mathbb{C}^{d_2}\right) \oplus \mathbb{C}^{d_3}\cong \mathbb{C}^{d_1} \oplus \left(\mathbb{C}^{2}\otimes \mathbb{C}^{d_2}\right) \oplus \mathbb{C}^{d_3}\text{,}
\end{equation}
where on the right-hand side of the above equation we have used a simple transformation into tensor products.

Given two representations, one may ask about matrices that link them in a structure preserving way similar to how a matrix between two vector spaces preserves linearity.

\begin{definition}[Intertwiner]
Let $\varPi_1:G\to \U(d_1)$ and $\varPi_2:G\to\U(d_2)$ be two representations of the group $G$. A matrix $A\in\mathbb{C}^{d_1\times d_2}$ is called an \textit{intertwiner} or \textit{$G$-homomorphism} if it has the following property for all group elements $g$:
\begin{equation}
    \varPi_1(g)A=A\varPi_2(g)\text{.}
\end{equation}
\end{definition}

There is a powerful lemma that gives the exact form of all such matrices in the case of irreducible representations:
\begin{lemma}[Schur's lemma]
Given two irreducible representations $\varPi_1:G\to\U(d_1)$ and $\varPi_2:G\to\U(d_2)$ of the group $G$, all the possible intertwiners between them are:
\begin{enumerate}
\item $\mathbb{0}$ if the two irreps are not equivalent,
\item $c\mathbb{1}$ for some $c\in\mathbb{C}$ if they are equivalent.
\end{enumerate}
\end{lemma}

One way to construct new representations from existing ones is by taking the tensor product of two representations:
\begin{align}
\begin{split}
    \varPi_1\otimes \varPi_2:G&\to \U(d_1) \otimes \U(d_2)\text{,}\\
    g&\mapsto \varPi_1(g)\otimes \varPi_2(g)\text{.}
\end{split}
\end{align}
These representations are usually not irreducible even if $\varPi_1$ and $\varPi_2$ were irreps. The decomposition of such representations into irreps is governed by so-called fusion rules, which are connected to the group structure.

Let us continue our discussion with the representation theory of the unitary group $\U(d)$:
\begin{proposition}
Every irreducible representation of the unitary group $\U(d)$ that appears in the irrep decomposition of the $t$-fold tensor product of the defining representation can be labeled by a Young diagram of $t$ boxes with at most $d$ rows. These are known as tensor irreducible representations (tensor irreps).
\end{proposition}

A Young diagram of $t$ boxes and $k$ rows is a collection of $t$ boxes arranged in rows and columns that are left-justified such that the number of boxes in each row is non-increasing. For example, the following a Young diagram of $13$ boxes in $5$ rows:
\begin{equation}\ytableausetup{boxsize=1em}
\ydiagram{4,3,3,2,1}
\end{equation}
It turns out that the label of the trivial representation is a Young diagram with zero boxes, whereas the label of the defining representation is a Young diagram with one box.

The fusion rules of the unitary group are governed by the following theorem:
\begin{theorem}
Let $\varPi_{\mu}$ and $\varPi_{\nu}$ be tensor irreps of $\U(d)$ labeled by Young diagrams $\mu$ and $\nu$ with at most $d$ rows (their depth is at most $d$). Their tensor product decomposes as
\begin{equation}\label{eq:lr-formula}
\varPi_{\mu}\otimes\varPi_{\nu}\cong \bigoplus_{\substack{\lambda \\\mathrm{depth}(\lambda)\leq d}} c^{\lambda}_{\mu,\nu} \varPi_{\lambda}\text{,}
\end{equation}
where $c^{\lambda}_{\mu,\nu}$ are nonnegative integers called Littlewood-Richardson coefficients.
\end{theorem}

These coefficients are well-known in the literature or can be calculated using the Littlewood-Richardson rule. The two basic examples are the following:
\begin{align}
    \varPi_{\Yt{1}}^{\otimes 2}&\cong \varPi_{\Yt{1,1}}\oplus \varPi_{\Yt{2}}\text{,}\\
     \varPi_{\Yt{1}}^{\otimes 3}&\cong \varPi_{\Yt{1,1,1}} \oplus \varPi_{\Yt{2,1}}^{\oplus 2} \oplus \varPi_{\Yt{3}}\text{.}
\end{align}

The main topic of this article is unitary $t$-designs, that is, finite collection of unitaries that reproduce a $t$-fold average (or twirl) of the entire unitary group. The twirl operation, in general, means the following:
\begin{definition}
    Let $G$ be a compact group, and $\varPi:G\to\U(d)$ be a unitary representation. Then, we define the twirl operation on a matrix $M\in\mathbb{C}^{d\times d}$:
    \begin{equation}
        \mathcal{T}(M)\coloneqq \int_{G} \varPi(g) M\varPi(g)^{\dagger} \mathrm{d}g\text{,}
    \end{equation}
    where the integration is over the Haar measure of the compact group $G$.
\end{definition}

For example in Eqs.~\eqref{eq:t_design_with_weights} and~\eqref{eq:twirl_by_G}, the twirling happens with the $t$-fold tensor product of the defining representation, that is, $U^{\otimes t}$.

A nice property of the twirl operation is that for any input matrix $M$ the output $\mathcal{T}(M)$ is a matrix that commutes with every representant operator $\varPi(g)$, or, put differently, a special kind of intertwiner. In the cases of unitary $t$-designs, the output commutes with every unitary of the form $U^{\otimes t}$.

If the representation $\varPi$ is irreducible, the output matrix is proportional to the identity due to Schur's lemma. The proportionality constant is $\Tr(M)$. That is, we have a 1-design.

\subsection{Recap of character theory}\label{app:character_theory}
    For a given finite group $G$, the representation theory is well explicable: There are exactly as many inequivalent irreps of $G$ as the number of conjugacy classes in $G$.
    Based on this notion, any representation of $G$ can be described by the multiplicities of the specific irreps.
    Usually, it is hard to compute the irreducible decomposition of a specific representation of $G$, to handle this problem, character theory is used.
    The character can be defined for any representation of a group in the following way.
    In the following, we shall use $\varPi$ and $\rho$, respectively, to differentiate between representations of the unitary group $\U(d)$ and the finite group $G$.
    
    \begin{definition}
        The \textit{character} of a representation $\rho$, denoted by $\chi_\rho$, is defined as:
        \begin{equation}
            \chi_\rho(g) \coloneqq \operatorname{Tr}(\rho(g)), \quad g \in G.
        \end{equation}
    \end{definition}
    The character $\chi_\rho$ is a \textit{class function} on $G$, because it is constant in each conjugacy class due to the cyclicity of the operator trace.
    A character $\chi_\rho$ is \textit{irreducible} if the representation $\rho$ is irreducible. 

    Characters encode significant information about the representations of any group (even infinite) and are useful in distinguishing inequivalent representations or checking the irreducibility of a representation.
    The irreducible characters of a finite group $G$ form a \textit{character table} which encodes useful information about $G$ in a compact form.
    The character table consists of $n$ rows and $n$ columns where $n$ is the number of conjugacy classes of $G$.
    Each row corresponds to an irreducible character and each column to a conjugacy class. The elements are the values of a specific irreducible character on a specific conjugacy class.
    The character table also carries orthogonality properties.

    The \textit{inner product} of two class functions $\alpha, \beta$ of $G$ is defined as
    \begin{equation}
        \left\langle \alpha, \beta \right\rangle \coloneqq \frac{1}{\left| G \right|} \sum_{g \in G} \alpha(g) \overline{\beta(g)}.
    \end{equation}

    \begin{proposition}
        [Orthogonality relation 1.]
        For any finite group $G$ and its irreducible characters $\chi_i,\chi_j$ the following holds:
        \begin{equation}
            \left\langle \chi_i, \chi_j \right\rangle = 
                \left\{\begin{matrix}
                    0 \;&\text{if}\enspace i \neq j,
                \\
                    1 \;&\text{if}\enspace i = j.
                \end{matrix}\right.
        \end{equation}
    \end{proposition}

    In our research, we investigate finite subgroups of the unitary group, with particular focus on their behavior under tensor powers. Given any finite group $G$ and a unitary representation $\rho$, the image of $\rho$ can be thought of as a subgroup of $\U(d)$. Therefore, in the following, we shall consider $G$ a subgroup of $\U(d)$ by $\rho$. Given the unitary group $\U(d)$, we consider the $t$-fold tensor power of its defining representation, and analyze how its irreducible components decompose when restricted to a finite subgroup $G < \U(d)$.
    For each irreducible component appearing in the tensor power representation (which are labeled by Young diagrams), we study the restriction to $G$ and determine its decomposition into irreps of $G$.
    To compute this decomposition, we utilize the character table of $G$ as well as its power map, which encodes how the powers of conjugacy classes map to other conjugacy classes. These tools allow us to determine precisely how a given representation of the full unitary group decomposes into irreducible components upon restriction to the finite subgroup.

    Let $G$ be a finite group with a representation $\rho$ in $d$ dimensions and $\gamma$ be a Young diagram with $t$ boxes in at most $d$ rows.
    The irreducible decomposition of the representation $\rho_\gamma$ of $G$ which is the restriction of $\varPi_\gamma$ to $\rho(G)$ can be calculated by taking the inner product of its character $\chi_\gamma(\cdot) = \operatorname{Tr}(\rho_\gamma(\cdot))$ with each irreducible character $\chi_j$.
    The inner product $\left\langle \chi_\gamma, \chi_j \right\rangle$ gives the multiplicity of the representation corresponding to $\chi_j$ in the irreducible decomposition of $\rho_\gamma$.
    The character $\chi_\gamma$ can be calculated as
    \begin{equation}
        \chi_\gamma (g) =\frac{1}{t!} \sum_{\pi \in \mathrm{S}_t} \gamma(\pi) \prod_{k=1}^{t} \chi_\rho(g^k)^{a_k(\pi)},
    \end{equation}
    where $\gamma(\pi)$ is the character of the irrep indexed by the Young diagram $\gamma$ of the $t$-degree symmetric group $\mathrm{S}_t$ evaluated at permutation $\pi$ and $a_k(\pi)$ is the number of cycles of length $k$ in $\pi$.
    This result is derived from Theorem 2.7.9 in Ref.~\cite{Lux_Pahlings_2010}.
    Using this equation, characters can be calculated using the GAP system for any specific group $G$ and the irreducibility of the representation $\rho_\gamma$ can be calculated.
    Thus, the conditions of Theorem~\ref{thm:creating_t_design_from_groups} can be checked and compatible groups $G_1,\ldots,G_\ell$ can be found.

\newpage

\bibliographystyle{quantum}
\bibliography{refs}

\end{document}